\documentclass[11pt]{amsart}
\usepackage[margin=1in]{geometry}
\usepackage{amsfonts, float, mathtools}
\usepackage{amssymb}
\usepackage[dvips]{graphics}
\usepackage{epsfig}
\usepackage{euscript}
\usepackage{color}
\usepackage{cite}

\usepackage[all]{xy}
\usepackage{multirow}
\usepackage{tikz}

\usepackage{fancyvrb}

\usepackage[pdftex]{hyperref}

 \newtheorem{thm}{Theorem}[section]
 
 \newtheorem{lemma}[thm]{Lemma}
 \newtheorem{prop}[thm]{Proposition}

 \theoremstyle{definition}
 
  \newtheorem{remk}[thm]{Remark}
 \theoremstyle{remark}
 
 \newtheorem{assumption}[thm]{Assumption}
 \numberwithin{equation}{section}
 
 \def\idtyty{{\mathchoice {\mathrm{1\mskip-4mu l}} {\mathrm{1\mskip-4mu l}} %
{\mathrm{1\mskip-4.5mu l}} {\mathrm{1\mskip-5mu l}}}}

 \def\idty{{\mathchoice {\mathrm{1\mskip-4mu l}} {\mathrm{1\mskip-4mu l}} %
{\mathrm{1\mskip-4.5mu l}} {\mathrm{1\mskip-5mu l}}}}

\newcommand{\bR}{{\mathbb R}}

\newcommand{\Ir}{{\mathbb Z}}

\newcommand{\bN}{{\mathbb N}}

\newcommand{\cA}{{\mathcal A}}
\newcommand{\cB}{{\mathcal B}}
\newcommand{\cF}{{\mathcal F}}
\newcommand{\cG}{{\mathcal G}}
\newcommand{\cH}{{\mathcal H}}

\newcommand{\cT}{{\mathcal T}}

\newcommand{\supp}{\operatorname{supp}}

\newcommand{\cP}{{\mathcal P}}
\newcommand{\cS}{{\mathcal S}}

\newcommand{\caA}{{\mathcal A}}

\newcommand{\caF}{{\mathcal F}}

\newcommand{\caH}{{\mathcal H}}
\newcommand{\caI}{{\mathcal I}}

\newcommand{\cK}{{\mathcal K}}

\newcommand{\caT}{{\mathcal T}}

\newcommand{\gap}{\mathrm{gap}}

\newcommand{\braket}[2]{\left\langle #1 , #2\right\rangle}

\newcommand{\id}{{\rm id}}
\newcommand{\ketbra}[1]{\vert #1\rangle\langle #1\vert}
\newcommand{\kettbra}[2]{\vert #1\rangle\langle #2\vert}

\newcommand{\dom}{\mathop{\rm dom}}
\newcommand{\spec}{\mathop{\rm spec}}

\newcommand{\be}{\begin{equation}}
\newcommand{\ee}{\end{equation}}
\newcommand{\bea}{\begin{eqnarray}}
\newcommand{\eea}{\end{eqnarray}}
\newcommand{\beann}{\begin{eqnarray*}}
\newcommand{\eeann}{\end{eqnarray*}}

\newcommand{\Rl}{\bR}

\newcommand{\eq}[1]{(\ref{#1})}

\title[Stability of the bulk gap]{Stability of the bulk gap for frustration-free topologically ordered quantum lattice systems}

\author[B. Nachtergaele]{Bruno Nachtergaele}
\address{Department of Mathematics and Center for Quantum Mathematics and Physics\\
University of California, Davis\\
Davis, CA 95616, USA}
\email{bxn@math.ucdavis.edu}

\author[R. Sims]{Robert Sims}
\address{Department of Mathematics \\
University of Arizona\\
Tuscon, AZ 85721, USA}
\email{rsims@math.arizona.edu}

\author[A. Young]{Amanda Young}
\address{
Munich Center for Quantum Science and Technology, and\\
Zentrum Mathematik, TU M\"{u}nchen\\
85747 Garching, Germany}
\email{young@ma.tum.de}

\begin{document}
\date{\today }
\begin{abstract}
We prove that uniformly small short-range perturbations do not close the bulk gap above the ground state of frustration-free quantum spin systems that satisfy a standard local topological quantum order condition. In contrast with earlier results, we do not require a positive lower bound for finite-system spectral gaps uniform in the system size. To obtain this result, we extend the Bravyi-Hastings-Michalakis strategy so it can be applied to perturbations of the GNS Hamiltonian of the infinite-system ground state.
\end{abstract}

\maketitle

%\tableofcontents

\section{Introduction}\label{sec:intro}

One of the characteristic properties of gapped topologically ordered ground state phases of quantum many-body systems 
is the stability of the spectral gap above the ground state with respect to small perturbations of the Hamiltonian. 
Stability results for the ground state gap have a long history. The first result that included the stability of the gap of the AKLT chain is due to Yarotsky \cite{yarotsky:2006}. 
The approach we follow in this paper has a much broader range of applicability; it was introduced by Bravyi, Hastings, and Michalakis \cite{bravyi:2010} and further developed in \cite{bravyi:2011,michalakis:2013,hastings:2019,nachtergaele:2022}.
Other approaches have been introduced in recent years \cite{de-roeck:2019, frohlich:2020, del-vecchio:2021, del-vecchio:2022}. 
These new approaches can also treat some cases of models with unbounded on-site Hamiltonians, see \cite[Section 1]{nachtergaele:2022} for a more detailed 
discussion. The Bravyi-Hastings-Michalakis strategy, however, is the only approach that 
handles general cases with non-trivial topological order.

One obstacle to proving spectral gaps for topological insulators is the common occurrence of gapless edge states. 
Spectral analysis for interacting many-body systems is usually carried out for finite systems for which edge states typically imply that there is no spectral gap uniform in the system size. Nevertheless, there may be a bulk gap, 
meaning excitations away from the boundary of the system have energy bounded below uniformly in the system size. The goal of this work is to prove stability for the bulk gap in a way that does not require the assumption of a uniform positive lower bound in the spectrum of finite systems. Previously, it was shown how certain cases can be handled by considering sequences of finite systems with suitable boundary conditions. For example, such an approach may 
work if the edge states are absent in the model considered with periodic boundary conditions \cite{michalakis:2013, nachtergaele:2022}.
In general, however, there may not be a suitable boundary condition that `gaps out' the boundary modes or we may 
not know whether such a boundary condition exists. Systems defined on a quasicrystal structure, for example, may
 be an instance where no simple way of removing gapless edge modes is available \cite{loring:2019}. In our approach here, we only assume that the infinite system described in the GNS representation of the ground state has a gap. 
 Under natural assumptions consistent with the current state-of-the-art for ground state stability, we prove that sufficiently small but extensive perturbations do not close the gap.

We adapt the strategy of Bravyi, Hastings, and Michalakis \cite{bravyi:2010, bravyi:2011,michalakis:2013} and use the techniques we developed in \cite{nachtergaele:2019,nachtergaele:2022} to handle the infinite system setting. From a certain perspective, and apart from the technical aspects to deal with
unbounded Hamiltonians, the infinite system setting  allows for a simplification in the statement of conditions and the main result.  In particular, the local topological quantum 
order (LTQO) condition is simpler to state directly for the infinite system. The LTQO property is well know to hold for one-dimensional systems with MPS ground 
states \cite[Appendix B]{nachtergaele:2022}. It is also well-established for Kitaev's quantum double models \cite{kitaev:2006,cui:2020} and the Levin-Wen string-net model 
\cite{levin:2005,qiu:2020a}. LTQO was recently also shown to hold for the AKLT model on a decorated hexagonal lattice \cite{lucia:2022}.

For concreteness, we work in the quantum spin system setting but, using the arguments of \cite{nachtergaele:2018}, our approach is applicable to lattice fermion systems too.

The assumption that the bulk Hamiltonian has a gap in the spectrum above the ground state appears in several 
important recent works.  For example, the construction of an index for the classification of symmetry-protected topological phases in the works of Ogata and  co-authors makes use of this assumption \cite{moon:2020,ogata:2021b,ogata:2020a,ogata:2021}. Other examples are in the recent work on adiabatic theorems for infinite many-body systems \cite{bachmann:2018b,bachmann:2021,henheik:2022}. All of these works use the same general setting as described here in Section \ref{sec:main}. In addition to the main stability result, we also prove Theorem \ref{thm:bounded_differentiability}, which shows that a differentiability assumption introduced in \cite{moon:2020} and also used in later work \cite{ogata:2019,ogata:2021e,ogata:2020a,ogata:2021} is always satisfied.

\section{Setup and statement of the main results}\label{sec:main}

\subsection{Setup and notation}\label{sec:setup}

The models considered in this work are defined on a $\nu$-regular discrete metric space $(\Gamma,d)$, for some $\nu>0$. This means that there exists $\kappa >0$ so that for all $x\in\Gamma, n\geq 1$, $|b_x(n)| \leq \kappa n^\nu$, where $b_x(n)= \{y\in \Gamma \mid d(x,y)\leq n\}$. For $\Lambda\in \cP_0(\Gamma)$, the finite subsets of $\Gamma$, and
$n\geq 0$, we also define the sets $\Lambda(n)$ by
\be\label{Lambda_ball}
\Lambda(n)=\bigcup_{x\in\Lambda} b_x(n).
\ee  
The algebra of local observables of the system is the usual
$\cA^{\rm loc} = \bigcup_{\Lambda\in\cP_0(\Gamma)} \cA_\Lambda$. Here, $\cA_\Lambda$ is the matrix algebra $\bigotimes_{x\in\Lambda} M_{d_x}$
with $d_x$ the dimension of the spin at $x$. The $C^*$-algebra of quasi-local observables $\cA$ is the completion of $\cA^{\rm loc}$ with respect to the operator norm. For $A\in\cA^{\rm loc}$, the {\em support} of $A$, 
denoted by $\supp A$, is the smallest $X\subset \Gamma$ such that $A\in \cA_X$. 
For any $X\subset\Gamma$, $\Pi_X: \cA\to\cA_X$ is the conditional expectation with respect to the tracial state $\rho$ on $\cA$:
\be\label{cond_exp_a}
\Pi_X = \rho\restriction_{\Gamma\setminus X} \otimes \id_{\cA_X}.
\ee
In particular, for local $A$, $\Pi_X(A)$ is a normalized partial trace.

We are specifically interested in systems defined on {\em infinite} $\Gamma$ and often want to consider
approximations $A_n\in \cA_{\Lambda_n}$ of $A\in\cA$, where $\Lambda_n\in \cP_0(\Gamma)$ 
is an increasing sequence of finite volumes such that $\bigcup_n \Lambda_n =\Gamma$. We call such
a sequence $(\Lambda_n)$ an {\em increasing and absorbing sequence} (IAS). 
It will often be important to have an estimate for the speed of convergence of $A_n\to A$,
in terms of a non-increasing function $g:[0,\infty)\to(0,\infty)$ that vanishes at infinity, which we call a {\em decay function}. 
In this paper we will only use decay functions that satisfy a moment condition of the form
\be\label{moment}
\sum_{n\geq 0} (n+1)^\mu g(n) < \infty,\ \mbox{  for some }\ \mu\geq \nu.
\ee
%In particular, the moment condition  implies that $g$ is bounded above by a power law:
%$$
%g(n) \leq \frac{C}{(n+1)^\mu}.
%$$
In particular, we will often work with decay functions of the form
\be \label{class_of_gs}
g(r) = \frac{1}{(1+r)^\xi} e^{-a r^\theta},\quad  \xi\geq 0, a\geq 0, \theta\in [0,1],
\ee
which share the property
\be\label{good_g}
g(n) g(m) \leq C g(n+m), \for{ for all } n,m\geq 0,
\ee
for some constant $C$. 

Consider an IAS $(\Lambda_n)$ in $(\Gamma,d)$, and a decay function $g$.
Define a norm $\Vert \cdot \Vert_{(\Lambda_n),g}$ on $\cA^{\rm loc}$ and a Banach space $\cA^{(\Lambda_n),g}$ by
\begin{equation} \label{g-norm}
	\Vert A\Vert_{(\Lambda_n),g} = \| A \| + \sup_{n \in \mathbb{N}} \left( \frac{\| A - \Pi_{\Lambda_n}(A) \| }{ g(n)} \right),\quad  \cA^{(\Lambda_n),g} =\overline{\cA^{\rm loc}}^{\Vert \cdot\Vert_{(\Lambda_n),g}}.
\end{equation}
For a proof that $\cA^{(\Lambda_n),g}$ is the Banach space of all $A\in\cA$ for which $\Vert A\Vert_{(\Lambda_n),g}<\infty$, see \cite{moon:2020}. In fact, $\cA^{(\Lambda_n),g}$ is a Banach $*$-algebra.

For each $x\in\Gamma$, $\Lambda_n: =b_x(n)$ defines a IAS. In this case we set $\Vert\cdot\Vert_{(b_x(n)),g} = \Vert\cdot\Vert_{x,g}$. Define the set
\be
\cA^g: =   \bigcup_{x\in\Gamma}   \cA^{(b_x(n)),g}.
\ee
For any decay function $g$ satisfying \eq{good_g},
any two norms from $\{\Vert \cdot \Vert_{x,g}\mid x\in\Gamma\}$ are equivalent. Hence, for all $x\in\Gamma$
\be
\cA^{(b_x(n)),g} = \cA^g. 
\ee
In this case, $\cA^g$ is a Banach $*$-algebra. Elements $A\in\cA^g$ are called {\em $g$-local}.

We will often also assume that a decay function $g$ is uniformly summable over $\Gamma$, i.e.,
\begin{equation} \label{uni_sum}
\Vert g\Vert_1 :=\sup_{x \in \Gamma} \sum_{y \in \Gamma} g(d(x,y)) < \infty \, ,
\end{equation}
and additionally, that there is a constant $C>0$ such that
\be \label{conv_const}
\sum_{z\in\Gamma} g(d(x,z))g(d(z,y)) \leq Cg(d(x,y)), \mbox{ for all } x,y\in \Gamma.
\ee
Any decay function $g$ satisfying (\ref{uni_sum}) and (\ref{conv_const}) will be called an $F$-function. For any $\nu$-regular $\Gamma$,
the following are $F$-functions appearing in this work:
\be \label{class_of_Fs}
F(r) = \frac{1}{(1+r)^\xi} e^{-a r^\theta},\quad  \xi>\nu+1, a\geq 0, \theta\in (0,1].
\ee
In the case $\Gamma=\Ir^\nu$, which is $\nu$-regular,  \eq{class_of_Fs}  defines an $F$-function for all $\xi>\nu$.
For a discussion of these examples and some basic inequalities, see \cite[Appendix]{nachtergaele:2019}.

\begin{assumption}[Initial Interaction]\label{ass:initial_model}
We assume the initial model is defined by a finite-range, uniformly bounded, frustration-free interaction $h$ given in terms of a family $h=\{ h_x \}_{x \in \Gamma}$
which satisfies: 
\begin{enumerate}
\item There is a number $R \geq 0$, called the interaction radius, for which
$h_x^* = h_x \in \mathcal{A}_{b_x(R)}$ for all $x \in \Gamma$. 
\item These terms are uniformly bounded in the sense that
\begin{equation} \label{uni_bd}
\| h \|_{\infty} = \sup_{x \in \Gamma} \| h_x \| < \infty \, .
\end{equation}
\item The interaction is {\it frustration-free}, meaning that $h_x \geq 0$ for all $x \in \Gamma$ and for any $\Lambda \in \mathcal{P}_0( \Gamma)$, 
\begin{equation}\label{HLambdan}
{\rm min \, spec}(H_{\Lambda}) = 0 \quad \mbox{where} \quad H_{\Lambda} = \sum_{\stackrel{x \in \Lambda:}{ {\rm supp}(h_x) \subset \Lambda}} h_x \, .
\end{equation}
\end{enumerate}
The frustration-free condition implies that the ground state space is ${\rm ker}(H_{\Lambda})$
for any finite volume $\Lambda$. Moreover, $\psi \in {\rm ker}(H_{\Lambda})$ if and 
only if $\psi \in {\rm ker}(h_x)$ for each $x \in \Gamma$ with ${\rm supp}(h_x) \subset \Lambda$. 
Thus, denoting by $P_{\Lambda}$ the orthogonal projection onto ${\rm ker}(H_{\Lambda})$, for any $\Lambda_0 \subset \Lambda$, one has
\begin{equation}
P_{\Lambda} P_{\Lambda_0} = P_{\Lambda_0} P_{\Lambda} = P_{\Lambda} \, .
\end{equation}
\end{assumption}

For such a model, the derivation $\delta_0$ determining the infinite system dynamics is given by
\be\label{delta0}
\delta_0(A) = \sum_{x\in\Lambda(R)} [ h_x, A]  \quad  \mbox{for any } A\in \cA_\Lambda \mbox{ and } \Lambda\in\cP_0(\Gamma).
\ee
It is a standard result that there is a closed derivation extending $\delta_0$, which we also denote by $\delta_0$,
with domain $\dom(\delta_0)$ for which $\cA^{\rm loc}$ is a core \cite[Theorem 6.2.4]{bratteli:1997} (note that the factor $i$ is absorbed in the definition 
of the derivation in this reference). The system dynamics
is then the strongly continuous one-parameter group of $C^*$-automorphisms $\{\tau^{(0)}_t \mid t\in\Rl\}$ satisfying
\be
\frac{d}{dt} \tau_t^{(0)} (A) = i\tau_t^{(0)} (\delta_0(A)) \quad \mbox{ for all } A\in \cA^{\rm loc}.
\label{dynamics}\ee
In fact, this differential equation holds for all $ A\in \dom (\delta_0)$. Two other general properties are:
\begin{enumerate}
\item $\tau_t^{(0)}(\dom(\delta_0))\subset \dom(\delta_0)$ for all $t \in \mathbb{R}$; 
\item $\tau_t^{(0)}(\delta_0(A)) =\delta_0(\tau_t^{(0)}(A))$ for all $A\in\dom(\delta_0)$ and $t \in \mathbb{R}$.
\end{enumerate}

More generally, quantum spin models can be defined by an interaction on $\Gamma$ which, by definition, is a map
$\Phi:\cP_0(\Gamma)\to\cA^{\rm loc}$, with the property that
$\Phi(X)^* = \Phi(X) \in\cA_X$ for all $X\in \cP_0(\Gamma)$. 
For any decay function $g$, an interaction norm is defined by  
\be
\Vert \Phi \Vert_g = \sup_{x,y\in\Gamma} g(d(x,y))^{-1} \sum_{\substack{X\in\cP_0(\Gamma): \\ x,y\in X}}\Vert \Phi(X)\Vert \, .
\ee
When the above quantity is finite for some interaction $\Phi$, the function $g$ is said to measure the decay of $\Phi$. 
If $g$ is an $F$-function, the norm $\Vert\cdot \Vert_g$ is called an {\em F-norm}.
If $g$ is summable, in the sense of (\ref{uni_sum}), and $\| \Phi \|_g < \infty$, 
then a closable derivation on $\cA^{\rm loc}$ can be defined by setting
\be\label{general_derivation}
\delta(A) = \sum_{Y,Y\cap X\neq \emptyset} [\Phi(Y),A] \quad \mbox{for }  A\in \cA^{\rm loc} \mbox{ with } {\rm supp}(A) \subset X  \in \mathcal{P}_0( \Gamma) \, .
\ee 
One can prove conditions that guarantee that the derivation $\delta$ defined on $\cA^{\rm loc}$ is a generator of a strongly continuous dynamics 
given by automorphisms of $\cA$ \cite{bratteli:1987,bratteli:1997}. In practice, however, one usually directly proves 
the existence of the thermodynamic limit of the Heisenberg dynamics $\{\tau_t\mid t\in\Rl\}$. 
%For a given interaction $\Phi$ and
%an IAS $( \Lambda_n)$, one first defines finite-volume Hamiltonians $H_{\Lambda_n}$ by 
%\begin{equation}
%H_{\Lambda_n} = \sum_{X \subset \Lambda_n} \Phi(X)
%\end{equation}
%and considers the corresponding Heisenberg dynamics:
%\be\label{finite-volume_dynamics}
%\tau_t^{\Lambda_n} (A) = e^{it H_{\Lambda_n}}  A  e^{-it H_{\Lambda_n}}  \quad \mbox{for } A \in \mathcal{A}_{\Lambda_n} \mbox{ and } t \in \mathbb{R} \, .
%\ee
%Next, one shows the convergence $\lim_n\tau_t^{\Lambda_n} (A) = \tau_t(A)$ for all $A\in \cA^{\rm loc}$;
%typically, these results are uniform for $t$ in compact intervals. Finally, one extends to all of 
%$\cA$ using the automorphism and group property of the limiting maps $\tau_t$. 
Standard results along these lines
prove the existence of the dynamics for $\Phi$ in a suitable Banach space of
interactions \cite{ruelle:1969,bratteli:1997,simon:1993} starting from a convergent series for small $|t|$. 
An alternative approach, based on Lieb-Robinson bounds \cite{lieb:1972}, was introduced by Robinson \cite{robinson:1976}. 
Lieb-Robinson bounds can be derived for any interaction $\Phi$ with a finite $F$-norm
 \cite{nachtergaele:2006}, and this allows one to extend the results for existence of the dynamics beyond
the Banach spaces of interactions $\cB_\lambda$ introduced by Ruelle \cite{ruelle:1969}. These ideas
are important for the construction of the spectral flow automorphisms \cite{bachmann:2012}.
This and some other generalizations relevant for the present work are discussed in detail in \cite{nachtergaele:2019}.

%In this work, we consider an infinite-volume, zero-energy, ground state of the frustration free model 
%defined by $\delta_0$, denoted by $\omega_0$. 
Recall that infinite-volume ground states associated to $\delta$ are those states $\omega$ on $\cA$ that satisfy
\be\label{defgs}
\omega (A^*\delta(A))\geq 0 \quad \mbox{for all } A\in \cA^{\rm loc}.
\ee
In the case of a frustration-free model as in \eq{HLambdan}, a  state $\omega$ is called a zero-energy ground state, or a frustration free ground state,
if  $\omega(h_x)=0$ for all $x\in\Gamma$. It is easy to see that a zero-energy ground state satisfies \eq{defgs}.

Let $(\cH,\pi,\Omega)$ be the GNS triple of $\omega$.
%
%This means that $\pi_0:\cA\to\cB(\cH)$ is a representation of the $C^*$-algebra $\cA$ on a Hilbert space $\cH$ for which $\{\pi_0(A) \Omega | A \in \cA^{\rm loc}\}$ is dense
%in $\cH$. Moreover, the normalized vector $\Omega \in \cH$ is such that $\omega_0(A) = \langle \Omega, \pi_0(A) \Omega\rangle$ for all $A\in \cA$. 
For the GNS representation of a ground state, as in (\ref{defgs}), 
there exists a unique, non-negative self-adjoint operator 
$H$ on $\cH$, with dense domain $\dom H$, satisfying $H \Omega =0$ and 
\be \label{rep_unpert_dyn}
\pi(\tau_t(A)) = e^{itH} \pi(A) e^{-itH} \quad \mbox{for all } A\in\cA \mbox{ and } t\in \Rl.
\ee
The full domain of $H$ is seldom described explicitly. However, for all systems we consider in this paper,
$\pi(\cA^{\rm loc})\Omega$ is a core for $H$.

%Frustration-free models are characterized by the property that they have 
%a zero-energy ground state, meaning there is a state $\omega_0$ on $\cA$
%such that  $\omega_0(h_x)=0$ for all $x\in\Gamma$. It is easy to see that a zero-energy ground state
%is a ground state in the sense of \eq{defgs}.
%Frustration-free models may have other ground states but thermodynamic limits of ground states of $H_{\Lambda_n} $
%%\be\label{HLambdan}
%%H_{\Lambda_n} = \sum_{\stackrel{x \in \Lambda_n:}{ {\rm supp}(h_x) \subset \Lambda_n}} h_x
%%\ee
%automatically have the zero-energy property. 
%Conversely, any zero-energy infinite volume ground state is a thermodynamic limit of finite-volume zero-energy ground states.
%%We will refer to the GNS representation of a zero-energy ground state of a frustration-free model as a {\em frustration-free representation}.
%
%The GNS Hamiltonian in a frustration-free representation takes the following explicit form:
%\be \label{inf_ham_sum}
%H_0 \psi = \lim_n \pi_0(H_{\Lambda_n}) \psi = \lim_n  \sum_{\stackrel{x \in \Lambda_n:}{ {\rm supp}(h_x) \subset \Lambda_n}} \pi_0(h_x) \psi \quad \mbox{for all } \psi \in \dom H_0
%\ee
%and any IAS $(\Lambda_n)$. This convergence follows directly from the fact that $\pi_0(\cA^{\rm loc})\Omega$ is a core for $H_0$.

The {\em (GNS) gap} of the model in the state $\omega$
is defined as
\be\label{bulkgap0}
{\rm gap}(H) = \sup \{ \gamma >0 \mid (0,\gamma)\cap \spec (H) =\emptyset\}.
\ee
If the set on the RHS is empty, one defines ${\rm gap}(H)=0$. We say that a ground state $\omega$ is gapped if  ${\rm gap}(H)>0$.

The equivalence of the following two conditions is easy to verify:
\begin{enumerate}
	\item For some $\gamma >0$, $\omega$ satisfies
	\begin{equation} \label{gap_check}
		\omega(A^* \delta(A)) \geq \gamma \omega(A^*A) \quad \mbox{for all } A \in \cA^{\rm loc} \mbox{ with } \omega(A) = 0;
	\end{equation}
	\item The ground state of the GNS Hamiltonian $H$ is unique and ${\rm gap}(H) \geq \gamma$. 
\end{enumerate}

A case of special interest is when $\Gamma$ is infinite and describes the bulk of a physical model while the same system on a subset of $\Gamma$ with a boundary would  describe an edge. In the first situation we will refer to the GNS gap as the {\em bulk gap} of the system.  A model with the same interaction restricted to a subspace of $\Gamma$ describing an edge, may have a vanishing gap while the bulk gap is positive. This is precisely the situation of interest here.

We will use that the GNS representation $\pi$ is an isometry. This follows from the fact that $\cA$ is simple \cite[Theorem 5.1]{glimm:1960}, which implies that $\ker\pi=\{0\}$.

\subsection{Main results}

We now state the assumptions for the main results. 

\begin{assumption}[Bulk gap]\label{ass:bulkgap}
We assume $\gamma_0:={\rm gap}(H_0)>0$, where $H_0$ is the GNS Hamiltonian of an infinite-volume, zero-energy ground state $\omega_0$ of a finite-range,
uniformly bounded, frustration free interaction  $\{h_x\}$ as in Assumption \ref{ass:initial_model}.
\end{assumption}

We also need to impose a condition that the local gaps do not close too fast. There generally is some 
freedom in choosing the family of finite volumes on which to impose this condition. We will assume that there is
a family $\cS = \{ \Lambda(x,n) \mid n\geq 0, x\in\Gamma\} \subset \cP_0(\Gamma)$, with $b_x(n)\subset \Lambda(x,n)$ for all $x$ and $n$, 
and an associated family of partitions of $\Gamma$ which \emph{separates $\cS$ and has at most polynomial growth}. Concretely, this means
there is a family of sets 
$\cT=\{\cT_n \mid n \geq 0\}$ and positive numbers $c$ and $\zeta$, such that for each 
$n\geq 0$, $\caT_n = \{\cT_n^i : i\in \caI_n \}$ is a partition of $\Gamma$ satisfying $|\caI_n|\leq c n^\zeta$ and
\be \label{sep_property}
\Lambda(x,n)\cap\Lambda(y,n) = \emptyset \;\text{  for all } \; x,y\in \cT_n^i \; \text{with}\; x \neq y.
\ee
In such cases, we say that $\cT$ is of $(c,\zeta)$-polynomial growth.

As an example, in the case of $\Gamma = \Ir^\nu$, we may take for $\Lambda(x,n)$ the $\ell^\infty$-ball of radius $n$ centered at $x$,
define $\caI_n = \Lambda(0,n)$ and, for each $i\in\caI_n$, set
\be
\cT^i_n=\{ x\in\Ir^\nu | x_j = i_j \mod 2n+1, i=1,\ldots, \nu\}.
\ee

\begin{assumption}[Local gaps]\label{ass:localgaps}
For an interaction $\{h_x\}$ of range $R$,  we assume there exist families $\cS$ and $\cT$, such that $\cT$
separates $\cS$  and is of $\zeta$-polynomial growth, and an exponent $\alpha\geq 0$ and constant $\gamma_1>0$, and
such that  the finite-volume Hamiltonians satisfy:
\be\label{localgaps}
\spec (H_{\Lambda(x,n)})\subset \{0\} \cup [n^{-\alpha} \gamma_1,\infty) \quad \mbox{for all } x \in \Gamma \mbox{ and } n \geq R.
\ee
\end{assumption}

It is important here that the local gaps are allowed to vanish in the limit of infinite system size. For example, certain types of topologically 
ordered two-dimensional systems are expected to have chiral edge modes with an energy of order $L^{-1}$ on a finite volume of diameter $L$. 
Whether or not such edge modes occur in frustration-free systems, however, is not clear. For the class of systems studied in \cite{lemm:2019b}, 
the authors find that finite-volume gaps of a system with gapless edge modes in the thermodynamic limit would have to decay at least as fast 
as $L^{-3/2}$. Other results of this type are in \cite{gosset:2016,kastoryano:2018,anshu:2020}.
This is consistent with the gapless boundary modes found in a class of toy models called Product Vacua with Boundary States which are of order $L^{-2}$ \cite{bachmann:2015,bishop:2016a}.
In any case, regardless of the possible values of the exponent $\alpha$, we will prove stability of the {\em bulk} gap.

The next assumption was introduced in the form we use here in \cite{michalakis:2013} 
where it is called {\em Local Topological Quantum Order} (LTQO). 

\begin{assumption}[LTQO]\label{ass:LTQO}
%Let $\omega_0$ be an infinite-volume, zero-energy ground state of a finite-range,
%uniformly bounded, frustration free model.
There is a decay function $G_0:[0,\infty)\to [0,\infty)$, with
\be\label{LTQO_moment}
\sum_{n\geq 0}n^q G_0(n)<\infty \mbox{ for some } q > 2(\nu + \zeta + \alpha),
%\sum_{n\geq 0} n^{2(\nu +1)} G_0(n) <  \infty,
\ee
and such that for all $m\geq k\geq 0, x\in\Gamma$, and $A \in \cA_{b_x(k)}$, the ground state projections satisfy
\be \label{LTQOforA}
\|P_{b_x(m)}AP_{b_x(m)}-\omega_0(A)P_{b_x(m)}\| \leq \|A\|  (1+k)^\nu G_0 (m-k).
\ee
\end{assumption}

As explained in detail in \cite[Section 8]{nachtergaele:2022},  if both the initial Hamiltonian and the 
perturbation (see below) have a local gauge symmetry, only observables $A$ that commute with this symmetry 
need to satisfy \eq{LTQOforA}. Other discrete symmetries can be treated similarly  (see
\cite[Section 8]{nachtergaele:2022}). Therefore, the stability results proved here (Theorems \ref{thm:finite_volume}
and \ref{thm:infinite_volume}) will also hold for symmetry-protected topological phases.

It is an interesting observation that the GNS Hamiltonians associated to frustration free models which satisfy Assumption~\ref{ass:LTQO} automatically
have a unique ground state. This is the content of the following propostion.

\begin{prop} \label{prop:unique_gs}
Let $\omega_0$ be an infinite-volume, zero-energy ground state of a frustration-free model satisfying Assumption~\ref{ass:LTQO}.
The kernel of the GNS Hamiltonian $H_0$ is one-dimensional.
\end{prop}
\begin{proof}
By way of contradiction, let us assume there is a unit vector $\psi \in {\rm ker}(H_0)$ with
$\langle \psi, \Omega \rangle =0$. To simplify notation let us denote by
$\tilde{A}:=\pi_0(A)$ the representative of  $A \in \cA$ in the GNS space.
The LTQO condition (\ref{LTQOforA}) can be restated in $\cB(\cH)$ as follows: for all $m \geq k \geq 0$, $x \in \Gamma$, and $A \in \cA_{b_x(k)}$, 
\begin{equation} \label{LTQO_in_GNS}
\| \tilde{P}_{b_x(m)} \tilde{A} \tilde{P}_{b_x(m)} - \langle \Omega, \tilde{A} \Omega \rangle \tilde{P}_{b_x(m)} \| \leq \| A \| (1+k)^{\nu} G_0(m-k)  \, .
\end{equation}

First, since $\pi_0( \cA^{\rm loc}) \Omega$ is dense in $\cH$, there is $x \in \Gamma$ and operators
$A_n \in \cA_{b_x(n)}$ for which $\psi_n = \tilde{A}_n \Omega \to \psi$. In this case, 
\begin{equation}
\lim_{n \to \infty} \langle \Omega, \tilde{A}_n \Omega \rangle = \lim_{n \to \infty} \langle \Omega, \psi_n \rangle = \langle \Omega, \psi \rangle = 0.
\end{equation}
Thus, there is $n_1 \in \mathbb{N}$ so that $| \langle \Omega, \tilde{A}_n \Omega \rangle| \leq 1/3$ for all $n \geq n_1$.

Next, since ${\rm ker}(H_0) \subset {\rm ker}( \tilde{H}_{b_x(m)})$ for all $m \geq 0$ and $x \in \Gamma$, we have
\begin{equation} \label{2_vecs_in_proj}
\langle \psi, \tilde{P}_{b_x(m)} \tilde{A}_n \tilde{P}_{b_x(m)} \Omega \rangle =\langle \psi, \tilde{A}_n \Omega \rangle =  \langle \psi, \psi_n \rangle \, 
\end{equation}
for all $m \geq 0$. Since $1 = \lim_{n \to \infty} \langle \psi, \psi_n \rangle$, there is some $n_2 \in \mathbb{N}$ for which 
$| \langle \psi, \psi_n \rangle| \geq 2/3$ for all $n \geq n_2$. Combining this with (\ref{2_vecs_in_proj}), we find that
for any $m\geq n \geq n_2$,
\begin{equation}
2/3 \leq \|  \tilde{P}_{b_x(m)} \tilde{A}_n \tilde{P}_{b_x(m)} \|. 
\end{equation}
It follows that for $n\geq \max\{n_1,n_2\}$,
\begin{equation}
1/3 \leq \|  \tilde{P}_{b_x(m)} \tilde{A}_n \tilde{P}_{b_x(m)} - \langle \Omega, \tilde{A}_n \Omega \rangle \tilde{P}_{b_x(m)} \|,
\end{equation}
a contradiction to \eq{LTQO_in_GNS}, as the decay function $G_0$ satisfies $G_0(m) \to 0$ as $m \to \infty$.
\end{proof}

Next, we turn to the perturbations of the Hamiltonian $H_0$. We consider 
$\Phi(x,n)^* = \Phi(x,n) \in \mathcal{A}_{b_x(n)}$ for all $x \in \Gamma$ and $n \geq 0$. 
These define what we call an {\em anchored} interaction $\Phi$.
By regrouping, we need only consider those terms with $n \geq R$.
% the interaction radius of $\{h_x\}$. 

\begin{assumption}[Short-range perturbation]\label{ass:short-range}
%Let $\Phi$ be an anchored interaction on $\Gamma$ with terms $\{ \Phi(x,n) \, | \, x \in \Gamma, \, n \geq 0 \, \}$. 
There is a constant $\| \Phi \| \geq 0$, $a >0$, and $\theta \in (0,1]$ such that for all $x\in\Gamma$
\be\label{pert_decay}
\Vert \Phi(x,n)\Vert \leq \Vert \Phi\Vert e^{-a n^\theta} \quad \mbox{ for all } n\geq R.
\ee
\end{assumption}

\begin{remk}
Assume $\Phi$ satisfies Assumption~\ref{ass:short-range}. 
As indicated in (\ref{class_of_Fs}), for any $0<a' <a$ and $\xi > \nu$, the function $F:[0, \infty) \to (0, \infty)$ given by
\begin{equation} \label{our_F}
F(r) = \frac{e^{- a' r^{\theta}}}{(1+r)^{\xi}} \quad \mbox{for all } r \geq 0,
\end{equation}
is an $F$-function on $\Gamma$. Let $\delta = a-a'>0$ and note that for any $x,y \in \Gamma$ with $d(x,y) \geq R$, we have 
\begin{eqnarray}
\sum_{z \in \Gamma} \sum_{\stackrel{n \geq R:}{x,y \in b_z(n)}} \| \Phi(z,n) \| & = & \sum_{n \geq d(x,y)} \sum_{z \in b_y(n) \cap b_x(n)} \| \Phi(z,n) \| \nonumber \\
& \leq & \kappa \| \Phi \| \sum_{n \geq d(x,y)} n^{\nu} e^{-a n^\theta} \leq C F(d(x,y))
\end{eqnarray}
with $C:= \kappa \| \Phi \| \sum_{n \geq 0} (1+ n)^{\nu + \xi} e^{- \delta n^{\theta}} < \infty$. Thus, $\| \Phi \|_F \leq C$, with $F$ as in (\ref{our_F}).
\end{remk}

The focus of this work is to analyze the stability of the bulk gap under the presence of perturbations given by an anchored interaction $\Phi$ satisfying Assumption~\ref{ass:short-range}. We consider perturbed Hamiltonians of the form
\be \label{fv_pert_ham}
H(\Lambda,s) = H_0 + s \pi_0(V_\Lambda), \quad s\in\bR
\ee
where, for any finite volume $\Lambda\in\cP_0(\Gamma)$,
\be\label{perturbation} 
V_{\Lambda} = \sum_{x \in \Lambda} \sum_{\stackrel{n\geq R:}{ b_x(n) \subset \Lambda}} \Phi(x,n).
\ee
Clearly, $V_\Lambda\in \cA_{\Lambda}$ is bounded and self-adjoint, and so $H(\Lambda,s)$ defines for all $s\in\Rl$ a self-adjoint Hamiltonian on $\cH$ with the same dense domain as $H_0$. 

In the next several sections we will prove the following theorem, which establishes that the spectral gap of $H(\Lambda,s)$ remains open for small $|s|$ uniformly in the finite volume $\Lambda$.

\begin{thm}[Stability of the gap uniformly in the perturbation region]\label{thm:finite_volume}
Suppose that $\{h_x\}$ and $\omega_0$ 
satisfy Assumptions \ref{ass:bulkgap} -- \ref{ass:LTQO}, and
$\Phi$ is an anchored interaction satisfying Assumption \ref{ass:short-range}. Then,
for all $\gamma \in (0,\gamma_0)$,  there exists $s_0(\gamma)>0$, such that for all real $s$, $|s| < s_0(\gamma)$, and
$\Lambda\in\cP_0(\Gamma)$, we have
\be\label{stable_gap}
\spec H(\Lambda,s) \subset \{ E(\Lambda,s) \} \cup [E(\Lambda,s)+\gamma,\infty)
\ee
with $H(\Lambda, s)$ as in (\ref{fv_pert_ham}) and
\be \label{pert_gse}
E(\Lambda , s) = \inf\spec H(\Lambda,s).
\ee
\end{thm}

We remark that the quantity  $s_0(\gamma)$ only depends on the values of $\kappa$ and $\nu$ 
of the lattice, $\| h \|$, the gap $\gamma_0$, the parameters in Assumption~\ref{ass:localgaps}, 
the decay function in Assumption~\ref{ass:LTQO}, and a suitable $F$-norm of the perturbation $\Phi$. 
From the arguments in this paper, one can derive an explicit lower bound for $s_0(\gamma)$ in terms 
of these quantities, see Section~\ref{sec:final_summary}.

We also investigate the situation where the perturbation region $\Lambda$ tends to all of $\Gamma$. 
Consider any IAS $(\Lambda_n)$. %and let $\Phi$ be an anchored interaction satisfying Assumption \ref{ass:short-range}. 
We will denote by $\tau_t^{(\Lambda_n,s)}$ the dynamics on $\cA$ generated by the derivation %$\delta_s^{\Lambda_n}$:
\begin{equation} \label{der_with_fv_pert}
\delta_s^{\Lambda_n}(A) = \delta_0(A) + [sV_{\Lambda_n}, A] \quad \mbox{for } A \in \cA^{\rm loc} \, .
\end{equation}
As discussed in \cite[Definition 3.7]{nachtergaele:2019}, the sequence of interactions  $h + s \Phi\restriction_{\Lambda_n}$
converges locally in $F$-norm to the interaction $h + s \Phi$. % here one may use any $F$-function with the form (\ref{our_F}).
Using \cite[Theorem 3.8]{nachtergaele:2019}, we conclude local convergence in the sense that
\be
\lim_{n \to \infty} \tau_t^{(\Lambda_n,s)}(A) = \tau_t^{(s)}(A) \quad \mbox{ for all } A\in\cA^{\rm loc}
\ee
as well as 
\be \label{limiting_der}
\lim_{n \to \infty} \delta_s^{\Lambda_n}(A) = \delta_s(A) \quad \mbox{ for all } A\in\cA^{\rm loc}
\ee
with $\tau_t^{(s)}$ (respectively, $\delta_s$) being the {\em a priori} well-defined strongly continuous dynamics on $\cA$ (respectively, the 
closure of the derivation restricted to $\cA^{\rm loc}$) generated by the interaction $h + s \Phi$.
Neither of these limits depend on the choice of IAS sequence $\Lambda_n$.

Our second result is then concerned with the ground state and its gap for a family of extensive perturbations. 
In particular, the uniformity of the stability result in Theorem~\ref{thm:finite_volume} allows one to prove, almost 
as a corollary, that for all $|s|\leq s_0(\gamma)$ there is a gapped ground state $\omega_s$ of $\delta_s$ in the 
sense of \eqref{gap_check}. To make this precise, we introduce the limiting spectral flow.
For any $\gamma >0$ and IAS $(\Lambda_n)$, take 
\be \label{limiting_spec_flow}
\alpha_s(A) = \lim_n \alpha^{\Lambda_n}_s(A) \quad \mbox{ for all } A\in\cA \, ,
\ee
where the spectral flows $\alpha_s^{\Lambda_n}$ will be introduced in more 
detail in the next section, see \eqref{spec_flow_on_A}. For now, it suffices to observe that this limit exists and is independent of the
choice of IAS. In fact, the interactions defining the spectral flows $\alpha_s^{\Lambda_n}$ converge 
locally in $F$-norm by arguments as in \cite[Section VI.E.2]{nachtergaele:2019}. This limiting spectral flow $\alpha_s$ defines a 
strongly continuous co-cycle of automorphisms of $\cA$, and moreover, under the assumptions 
we have made, for $A\in\cA^{\rm loc}$, $s \mapsto \alpha_s(A)$ is differentiable to all orders. 
We prove bounded differentiability for $A\in\cA^g$, for suitable $g$ in Theorem \ref{thm:bounded_differentiability}.

\begin{thm}[Stability of the bulk gap]\label{thm:infinite_volume}
Under the assumptions of Theorem \ref{thm:finite_volume}, let $\gamma \in (0,\gamma_0)$ and 
take $s$ with $|s| < s_0(\gamma)$. %the quantity whose existence is guaranteed by Theorem \ref{thm:finite_volume}.
%For the limiting spectral flow $\alpha_s$, as in (\ref{limiting_spec_flow}),
The state $\omega_s=\omega_0\circ\alpha_s$ is a gapped ground state of the perturbed infinite dynamics $\delta_s$, i.e.
\be\label{perturbedGNSgap}
\omega_s (A^* \delta_s (A)) \geq \gamma \omega_s (A^*A)  \quad \mbox{ for all } A\in\cA^{\rm loc} \mbox{ with } \omega_s(A) = 0\, .
\ee
In particular, the GNS Hamiltonian $H_s$ of $\omega_s$ has a one-dimensional kernel and $\spec H_s$ has a gap above its ground state bounded below by $\gamma$.
\end{thm}

\section{Quasi-locality, Domains and Local Decompositions} \label{sec:ql_d_ld}

The strategy used here for proving spectral gap stability of infinite systems relies in an essential way 
on quasi-locality properties of the observables, the dynamics, and several transformations defined in
terms of the dynamics. Quasi-locality of observables is the topic of Section~\ref{sec:ql}. In Section~\ref{sec:qlm_gf}, we recall general methods for making strictly local approximations of both quasi-local observables and maps. The specific quasi-local maps and estimates used in the stability proof are discussed in Sections~\ref{sec:un_pert_dyn} - \ref{sec:spec_flow}. Their corresponding counterparts in the GNS representation are treated in Section~\ref{sec:gns}.
In Section~\ref{sec:domain}, we prove how the action of certain unbounded operators on a dense domain can expressed as limits of sequences of bounded operators with finite support.
In Section \ref{sec:bounded_differentiability} we prove a local bound on the generator of the quasi-adiabatic evolution which, while not used in this work, is important for applications.

\subsection{Quasi-Locality} \label{sec:ql}

We first recall some general features of quasi-locality estimates and then turn to some important examples relevant for this work.

\subsubsection{Quasi-locality estimates.} \label{sec:qlm_gf}

If for $A \in \cA$, $X \in \mathcal{P}_0( \Gamma)$, and $\epsilon >0$
\begin{equation}
\|[A,B]\| \leq \epsilon \|B\|\, , \quad \mbox{whenever } B\in \cA_{\Gamma\setminus X}^{\rm loc} \, ,
\end{equation}
the following estimate was shown  in \cite{bravyi:2006a,nachtergaele:2013}:
\be \label{gen_ql_est}
\|A - \Pi_X(A)\| \leq \epsilon \, .
\ee

A linear map $\cK:\cA\to \cA$ is said to be {\it quasi-local} with constant $C\geq0$, power $p\geq 0$, and decay 
function $G$ if
\be\label{def:qlm}
\|[\cK(A),B]\| \leq C|X|^p\|A\|\|B\|G(d(X,Y)) \mbox{ for all }A\in \cA_X \mbox{ and } B\in \cA_Y \, . 
\ee
%This definition $\cK$, the image $\cK(A)$ of a strictly local
%observable $A$ can be approximated by a strictly local observable using (\ref{gen_ql_est}).
%In the present context, it suffices to localize observables on balls.
%For any quasi-local map $\cK$ on $\cA$ satisfying (\ref{def:qlm}),
Using (\ref{gen_ql_est}), for such a map $\cK$ and $A\in\cA_{b_x(k)}$,  we have
\be\label{cond_exp_est}
\|\cK(A) - \Pi_{b_x(k+n)}(\cK(A))\| \leq C |b_x(k)|^p \|A\|G(n).
\ee
When the corresponding decay function $G$ is summable, this estimate guarantees the absolute convergence of telescopic sums, i.e. for any $n_0 \geq 0$, 
\begin{equation}
\cK(A) = \Pi_{b_x(k+n_0)}( \cK(A)) + \sum_{n = n_0+1}^{\infty} \left( \Pi_{b_x(k+n)} - \Pi_{b_x(k+n-1)} \right)( \cK(A)) 
\end{equation}
since the terms satisfy 
\begin{equation} \label{ce_diff_bd}
\|  \left( \Pi_{b_x(k+n)} - \Pi_{b_x(k+n-1)} \right)( \cK(A)) \| \leq 2 C |b_x(k)|^p \|A\|G(n-1) \quad \mbox{for } n \geq 1.
\end{equation}
A common choice is $n_0=0$ and we adopt the notation
\be\label{delta_ops}
\cK(A) = \sum_{n\geq k} \Delta_{b_x(k)}^n(\cK(A)) \quad \text{where} \quad \Delta_{b_x(k)}^n = \begin{cases} \Pi_{b_x(k)}, & n = k \\ \Pi_{{b_x(n)}}-\Pi_{{b_x(n-1)}}, & n\geq k+1. \end{cases}
\ee

%\subsubsection{Four quasi-local maps} \label{sec:ql_ex}
We now review four quasi-local maps as well as a few of their important properties that are used in the stability argument. For more details of these maps see \cite{nachtergaele:2019}. Throughout Sections \ref{sec:un_pert_dyn}--\ref{sec:gns}  we work under the assumptions of Theorem \ref{thm:finite_volume}.

\subsubsection{Dynamics} \label{sec:un_pert_dyn}

It is well-known that the unperturbed dynamics  $\tau_t^{(0)}$ defined as in \eq{dynamics} 
satisfies an exponential Lieb-Robinson bound \cite{lieb:1972}. Namely, for every $\mu >0$ there exists $C_{\mu} >0$ and $v_{\mu}>0$ such that the bound
\begin{equation} \label{unpert_lrb}
	\| [ \tau_t^{(0)}(A), B ] \| \leq C_{\mu} \min(|X|,|Y|) \| A \| \| B \| e^{- \mu(d(X,Y) - v_\mu |t|)} 
\end{equation}
holds for any $X,Y \in \mathcal{P}_0(\Gamma)$, all $A \in \mathcal{A}_X$, $B \in \mathcal{A}_Y$, and $t \in \mathbb{R}$,

It is easy to check that the perturbed interaction $h + s \Phi{\restriction_\Lambda}$
has a finite $F$-norm for the same $F$ as $\Phi$,  and that this $F$-norm is uniformly bounded in $|s|\leq 1$ and $\Lambda$.
As a consequence, there are $C_F>0$ and $v_F>0$, independent of $s$ and $\Lambda$,
such that for any choice of $X,Y \in \mathcal{P}_0(\Gamma)$, 
\begin{equation} \label{pert_dyn_lrb}
\| [ \tau_t^{(\Lambda, s)}(A), B ] \| \leq C_F \| A \| \| B \| e^{v_F |t|} \sum_{x \in X} \sum_{y \in Y} F(d(x,y))  
\end{equation}
for all $A \in \mathcal{A}_X$, $B \in \mathcal{A}_Y$, and $t \in \mathbb{R}$.

%It will also be useful for us to observe that this perturbed dynamics $\tau_t^{(\Lambda, s)}$ may 
%equivalently be viewed as the dynamics generated by the closure of the 
%derivation $\delta_s^{\Lambda}$ defined by $\cA^{\rm loc}$ by
%\begin{equation} \label{def_der_pert_dyn}
%\delta_s^{\Lambda}(A) = \delta_0(A) + [s V_{\Lambda}, A] \quad \mbox{for all } A \in \cA^{\rm loc}  . 
%\end{equation}
Since each $sV_{\Lambda}$ is bounded and self-adjoint, \cite[Proposition 5.4.1]{bratteli:1997} implies that
\begin{equation} \label{int_pic_dyn}
\tau_t^{(\Lambda, s)}(A) = (K_t^{(\Lambda, s)})^* \tau_t^{(0)}(A) K_t^{(\Lambda, s)} \quad \mbox{for all } A \in \cA \mbox{ and } t \in \mathbb{R}
\end{equation}
where $\{ K_t^{(\Lambda, s)} \, | \, t \in \mathbb{R} \}$ is a one-parameter family of unitaries on $\cA$ which are uniquely defined as the
$\cA$-valued solution of 
\begin{equation} \label{def_int_pic_unis}
\frac{d}{dt} K_t^{(\Lambda, s)} = -i \tau_t^{(0)}(sV_{\Lambda}) K_t^{(\Lambda, s)} \quad \mbox{with} \quad K_0^{(\Lambda, s)} = \idty \, .
\end{equation} 
These unitaries are quasi-local as, for any $A \in \cA^{\rm loc}$ and $t>0$,
\begin{eqnarray}
\| [ K_t^{(\Lambda, s)}, A] \| = \| (K_t^{(\Lambda, s)})^*A K_t^{(\Lambda, s)} - A \| & = & \left\| \int_0^t \frac{d}{du} \left( (K_u^{(\Lambda, s)})^*A K_u^{(\Lambda, s)}\right) \, du \right\| \,  \nonumber \\
& \leq & \int_0^t \| [ \tau_u^{(0)}(sV_{\Lambda}), A] \| \, du  \, .
\end{eqnarray}
An application of (\ref{unpert_lrb}) then shows that for any $\mu >0$ and $A \in \cA_X$ with $X \in \mathcal{P}_0( \Gamma \setminus \Lambda)$, 
\begin{equation} \label{int_pic_uni_ql_est}
\| [ K_t^{(\Lambda, s)}, A] \| \leq C_{\mu} |s| | \Lambda| \| V_{\Lambda} \| \| A \| |t| e^{\mu v_{\mu} |t|} e^{- \mu d(\Lambda, X)}
\end{equation}
for any $s,t \in \mathbb{R}$. Thus, $K_t^{(\Lambda, s)} \in \cA^g$ for any exponential $g$, by (\ref{gen_ql_est}).

\subsubsection{Weighted Integral Operators} 

Fix $\gamma >0$. For each $\Lambda \in \mathcal{P}_0( \Gamma)$ and $s\in\bR$, we define two
weighted integral operators $\cF_s^{\Lambda} : \cA \to \cA$ and $\cG_s^{\Lambda} : \cA \to \cA$ by
\begin{equation} \label{def:F+G}
\cF_s^{\Lambda}(A) = \int_{- \infty}^{\infty} \tau_t^{(\Lambda, s)}(A) w_{\gamma}(t) \, dt \quad \mbox{and} \quad \cG_s^{\Lambda}(A) = \int_{- \infty}^{\infty} \tau_t^{(\Lambda, s)}(A) W_{\gamma}(t) \, dt \, ,
\end{equation}
where the real-valued functions $w_{\gamma}, W_{\gamma}\in L^1(\mathbb{R})$, are defined in \cite[Section VI.B]{nachtergaele:2019}.
In particular, they decay faster than any stretched exponential.
%By strong-continuity of the perturbed dynamics, both of the integrals appearing in (\ref{def:F+G}) are well-defined in the sense of Bochner, see e.g. \cite{yosida:1980}. 
Both of these maps depend on the choice of $\gamma$ through their weight functions, $w_\gamma$ and $W_\gamma$ respectively, but
we suppress this in the notation. Arguing as in 
\cite[Section VI.E.1]{nachtergaele:2019}, see also \cite[Section 4.3.2]{nachtergaele:2022}, we find that for all $A \in \cA$
\begin{equation} \label{wio_op_bd}
\| \cF_s^{\Lambda}(A) \| \leq \| A \| \quad \mbox{and} \quad \| \cG_s^{\Lambda}(A) \| \leq \| W_\gamma \|_1 \| A \| \, ,
\end{equation}
i.e. these maps are bounded uniformly with respect to $s \in \mathbb{R}$ and $\Lambda \in \mathcal{P}_0( \Gamma)$. 
Moreover, they are uniformly quasi-local in the sense that for each $\cK \in \{ \cF, \cG\}$ there is a
decay function $G_{\cK}$ such that: for any choice of $X,Y \in \mathcal{P}_0(\Gamma)$, we have 
\begin{equation} \label{wio_gen_ql_est}
\sup_{s \in [-1,1]} \| [ \cK_s^{\Lambda}(A), B ] \| \leq 2 \| A \| \| B \| |X| G_{\cK}(d(X,Y))  
\end{equation}
for all $A \in \mathcal{A}_X$ and $B \in \mathcal{A}_Y$. As shown in \cite[Lemma 6.10--6.11]{nachtergaele:2019}, the decay functions
$G_{\cK}$ can be made explicit. For our purposes here, we need only stress that they can be taken independent of
$\Lambda \in \mathcal{P}_0( \Gamma)$ and $s \in [-1,1]$, and with decay faster than any power. Thus, 
for any $\mu \geq 0$, 
\begin{equation} \label{all_finite_moms}
\sum_{n=1}^{\infty} (n+1)^{\mu} G_{\cK}(n) < \infty \, .
\end{equation}

\subsubsection{The Spectral Flow} \label{sec:spec_flow}

Fix $\gamma >0$. For each $\Lambda \in \mathcal{P}_0( \Gamma)$ and $s\in\bR$, denote by 
\be\label{SF_gen_alg}
D(\Lambda,s)= \cG_s^{\Lambda}(V_{\Lambda}) = \int_{-\infty}^\infty \tau^{(\Lambda, s)}_t (V_\Lambda) W_\gamma(t)   dt
\ee
with $\cG^{\Lambda}_s$ as defined in (\ref{def:F+G}). 
Clearly, $D(\Lambda, s)$ is self-adjoint and $s \mapsto D(\Lambda, s)$ is uniformly bounded by (\ref{wio_op_bd}). 

For $t \in \mathbb{R}$ fixed, the 
strong derivative of $s \mapsto \tau_t^{(\Lambda, s)}$ is given by the Duhamel formula  \cite[Proposition 2.7]{nachtergaele:2019}:
\begin{equation} \label{para_der}
\frac{d}{ds} \tau_t^{(\Lambda, s)}(A) = i \int_0^t \tau_r^{(\Lambda, s)}([V_{\Lambda}, \tau_{t-r}^{(\Lambda, s)}(A) ]) \, dr \, .
\end{equation}
Using (\ref{para_der}), one obtains the norm continuity of $s \mapsto D(\Lambda, s)$ from the following estimate:
\begin{eqnarray} \label{para_cont_est}
\| D(\Lambda, s) - D(\Lambda, s_0) \| & \leq & \int_{-\infty}^{\infty} \| \tau^{(\Lambda, s)}_t (V_\Lambda)  - \tau^{(\Lambda, s_0)}_t (V_\Lambda) \| | W_\gamma(t)|   dt \nonumber \\
& \leq & 2 \| V_{\Lambda} \|^2 |s-s_0| \int_{- \infty}^{\infty} |t| | W_{\gamma}(t)| \, dt .
\end{eqnarray}
%Again we refer to \cite[Section VI.B]{nachtergaele:2019} for details.

Given these properties of $D(\Lambda, s)$, there is a unique solution of 
\be\label{s_f_eq}
\frac{d}{ds}  U(\Lambda,s) = -i D(\Lambda,s)  U(\Lambda,s) \quad \mbox{with} \quad U(\Lambda,0) = \idty \, ,
\ee
which is given by unitaries in $\cA$. %In fact, as in the proof of (\ref{int_pic_dyn}), this solution forms a one-parameter family of unitaries in $\cA$. 
Using similar arguments as in \eq{int_pic_uni_ql_est} with (\ref{wio_gen_ql_est}) and (\ref{SF_gen_alg}), one can  show that for $s>0$
\begin{equation} \label{spec_flow_uni_ql_est}
\| [ U(\Lambda,s), A ] \| \leq \int_0^s \| [ \cG_r^{\Lambda}(V_{\Lambda}), A] \| \, dr  \leq 2 s \| A \| | \Lambda | \| V_{\Lambda} \| G_{\cG}(d(X, \Lambda)) 
\end{equation}
for any $A \in \cA_X$ with $X \in \mathcal{P}_0( \Gamma \setminus \Lambda)$. 
Thus, $U(\Lambda,s) \in \cA^g$ for some $g$ with finite moments of all orders by \eq{all_finite_moms}.

The {\em spectral flow} is then the family of inner automorphisms on $\cA$ induced by $U(\Lambda,s)$:
\be \label{spec_flow_on_A}
\alpha^\Lambda_s(A) = U(\Lambda,s)^* A U(\Lambda,s) \quad \mbox{ for all } A\in \cA.
\ee
Quasi-locality of this map is then a consequnece of a Lieb-Robinson bound.  To this end, first rewrite the generator as 
\begin{equation}
D(\Lambda, s) = \cG_s^{\Lambda}(V_{\Lambda}) = \sum_{x \in \Lambda} \sum_{\stackrel{n \geq R:}{ b_x(n) \subset \Lambda}} \cG_s^{\Lambda}(\Phi(x,n)) \, 
\end{equation}
using (\ref{perturbation}). Applying the conditional expectations and telescopic sum from (\ref{delta_ops}), we further write  
\begin{equation} \label{def_spec_flow_int}
D(\Lambda, s) = \sum_{x \in \Lambda} \sum_{m \geq R} \Psi_{\Lambda}(x,m,s) \quad \mbox{where} \quad 
\Psi_{\Lambda}(x,m,s) = \sum_{\stackrel{R \leq n \leq m:}{b_x(n) \subset \Lambda }} \Delta_{b_x(n)}^{m}( \cG_s^{\Lambda}( \Phi(x,n))) \, .
\end{equation}
Arguing as in \cite[Appendix A]{nachtergaele:2022}, there is a decay function $G_{\Psi}$ and a positive number $\| \Psi \|_{G_{\Psi}}$ such that 
for all $\Lambda \in \mathcal{P}_0( \Gamma)$, $s \in[-1,1]$, $x \in \Lambda$ and $k \geq R$,
\begin{equation} \label{est_spec_flow_int}
\sum_{m \geq k} \| \Psi_{\Lambda}(x,m,s) \| \leq \| \Psi \|_{G_{\Psi}} G_{\Psi}(k) \, .
\end{equation}
One can be explicit about estimates for $G_{\Psi}$, see  \cite[Corollary A.3]{nachtergaele:2022}, but for our purposes,
we only need that is has finite moments of all orders.
Given (\ref{def_spec_flow_int}) and (\ref{est_spec_flow_int}), well-known Lieb-Robinson
bounds imply the existence of a decay function $G_{\alpha}$ so that
for all $X,Y \in \mathcal{P}_0(\Gamma)$,
\begin{equation}
\| [ \alpha^{\Lambda}_s(A), B ] \| \leq s  \| A \| \| B \| |X| G_{\alpha}(d(X,Y))  
\end{equation}
for all $A \in \mathcal{A}_X$, $B \in \mathcal{A}_Y$, and $s \in \mathbb{R}$.  
$G_{\alpha}$ is independent of $\Lambda \in \mathcal{P}_0(\Gamma)$ and has finite moments of all orders.

\subsection{In the GNS space} \label{sec:gns}

The spectral perturbation arguments are carried out in the GNS representation of the reference state $\omega_0$. 
The quasi-local maps discussed in the previous subsection
can be lifted to the GNS representation and we now present the necessary properties we will need in this setting.

\subsubsection{Dynamics}

As discussed in Section \ref{sec:setup}, the unperturbed dynamics $\tau_t^{(0)}$ is implemented in the GNS representation of $\omega_0$
by the GNS Hamiltonian $H_0$, as in \eq{rep_unpert_dyn}. We further show that the perturbed dynamics $\tau_t^{(\Lambda,s)}$
is implemented in the GNS representation of $\omega_0$ by the Hamiltonian
\begin{equation} \label{fv_pert_in_gns}
H(\Lambda, s) = H_0 + \pi_0(sV_{\Lambda}), s\in\Rl,
\end{equation}
from \eq{fv_pert_ham}. Specifically,
\begin{equation} \label{equal_pert_dyn}
\pi_0(\tau_t^{(\Lambda, s)}(A)) = e^{it H(\Lambda,s)} \pi_0 (A) e^{-it H(\Lambda,s)} \quad \mbox{for all } A \in \cA \mbox{ and } t \in \mathbb{R} \, . 
\end{equation}
Applying the GNS representation to the interaction picture representation (\ref{int_pic_dyn}) gives
\begin{equation}
\pi_0(\tau_t^{(\Lambda, s)}(A)) = (\pi_0(K_t^{(\Lambda, s)}))^* e^{itH_0} \pi_0(A) e^{-itH_0} \pi_0(K_t^{(\Lambda, s)}) \, .
\end{equation}
Then, (\ref{equal_pert_dyn}) follows by observing that 
\begin{equation}
e^{-itH(\Lambda, s)} = e^{-itH_0} \pi_0(K_t^{(\Lambda, s)}) \, ,
\end{equation}
as by \eq{def_int_pic_unis} $\tilde K_t^{(\Lambda, s)}:= \pi_0(K_t^{(\Lambda, s)})$ is the unique, unitary solution of
\begin{equation}
\frac{d}{dt} \tilde{K}_t^{(\Lambda, s)} = - i e^{itH_0} \pi_0(sV_{\Lambda})e^{-itH_0} \tilde{K}_t^{(\Lambda, s)} \quad \mbox{with} \quad \tilde{K}_0^{(\Lambda, s)} = \idty.
\end{equation}

\subsubsection{Weighted Integral Operators} 

For any $\gamma >0$, $\Lambda \in \mathcal{P}_0( \Gamma)$, and $s \in \mathbb{R}$ 
we map the weighted integral operators of (\ref{def:F+G}) to the GNS space by defining $\tilde{\cF}_s^{\Lambda}$ and
$\tilde{\cG}_s^{\Lambda}$ by
\begin{equation} \label{wio_in_GNS}
\tilde{\cF}_s^{\Lambda}(A) = \int_{-\infty}^{\infty} e^{itH(\Lambda, s)} A e^{-itH(\Lambda, s)} w_{\gamma}(t) \, dt \quad \mbox{and} \quad 
\tilde{\cG}_s^{\Lambda}(A) = \int_{-\infty}^{\infty} e^{itH(\Lambda, s)} A e^{-itH(\Lambda, s)} W_{\gamma}(t) \, dt 
\end{equation}
for all $A \in \cB(\cH)$. Using (\ref{equal_pert_dyn}), it is clear that 
\begin{equation} \label{reps_of_wios}
\pi_0( \cF_s^{\Lambda}(A)) = \tilde{\cF}_s^{\Lambda}(\pi_0(A)) \quad \mbox{and} \quad \pi_0( \cG_s^{\Lambda}(A)) = \tilde{\cG}_s^{\Lambda}(\pi_0(A)) \quad
\mbox{for all } A \in \cA\, .
\end{equation}

\subsubsection{The Spectral Flow} 

For fixed $\Lambda\in\cP_0(\Gamma)$, following \cite[Section VI.A]{nachtergaele:2019} we define a norm-continuous family of unitaries $\tilde{U}(\Lambda,s)\in\cB(\cH)$
as the unique solution of
\begin{equation}\label{tildeU}
	\frac{d}{ds}\tilde{U}(\Lambda, s) = - i \tilde{D}(\Lambda, s) \tilde{U}(\Lambda, s) \quad \mbox{with} \quad \tilde{U}(\Lambda, 0) = \idty \, ,
\end{equation}
where
\begin{equation}
	\tilde{D}(\Lambda, s)  = \int_{-\infty}^{\infty} e^{itH(\Lambda, s)} \pi_0(V_\Lambda)e^{-itH(\Lambda, s)} W_\gamma(t) \, dt\, . % = \tilde{\cG}_s^{\Lambda}(\pi_0(V_\Lambda)) \, .
\end{equation}
%Comparing this unique solution with the representative of the unique solution of (\ref{s_f_eq}), we conclude that
%$\hat{\alpha}_s^{\Lambda} = \tilde{\alpha}_s^{\Lambda}$ as defined in (\ref{def_spec_flow_bh}).  
The {\em spectral flow} associated with $H(\Lambda,s)$ is the family of automorphisms of $\cB(\cH)$ defined by
\begin{equation} \label{def_spec_flow_bh}
	\tilde{\alpha}_s^{\Lambda}(A) = \tilde{U}(\Lambda,s)^* A \tilde{U}(\Lambda,s) .
\end{equation}
This is Hastings' quasi-adiabatic evolution \cite{hastings:2004,hastings:2005}.

By \eq{reps_of_wios} it is clear that 	$\tilde{D}(\Lambda, s)=\pi_0(D(\Lambda, s))$ with $D(\Lambda,s)$ as in \eq{SF_gen_alg} and, hence, by the uniqueness of the unitary solution 
of \eq{tildeU}, $\tilde{U}(\Lambda,s) = \pi_0(U(\Lambda,s))$, where $U(\Lambda,s)$ is as in (\ref{s_f_eq}). Therefore, $\pi_0$ lifts  the spectral flow in $\cA$ to the GNS space:
\begin{equation} \label{rep_spec_flow}
	\pi_0( \alpha_s^{\Lambda}(A)) = \tilde{\alpha}_s^{\Lambda}( \pi_0(A)) \quad \mbox{for all } A \in \cA \, .
\end{equation}

Recall that $E(\Lambda,s)$ denotes the ground state energy of $H(\Lambda,s)$.
Under our assumptions the ground state space of $H(\Lambda,0)=H_0$ is one-dimensional by
Proposition~\ref{prop:unique_gs},  and $\gamma_0:=\gap  (H_0)$ is strictly positive.
By standard results, see e.g \cite{kato:1995}, for $|s|$ sufficiently small the kernel of $H(\Lambda, s) - E(\Lambda, s) \idty$
is one-dimensional and the ground state gap does not immediately close. More precisely, for any $\gamma\in(0,\gamma_0)$, 
there is $s_0^{\Lambda}(\gamma) >0$ so that
\begin{equation} \label{gap_constraint}
{\rm gap}(H(\Lambda, s)) := \sup\{ \delta >0 : (E(\Lambda, s), E(\Lambda, s) + \delta) \cap {\rm spec}(H(\Lambda, s)) = \emptyset \} \geq\gamma,
\end{equation}
for all $|s| \leq s_0^\Lambda(\gamma)$.
An application of \cite[Theorem 6.3]{nachtergaele:2019} shows that
\begin{equation} \label{gen_spec_flow_proj_bh}
\tilde{\alpha}_s^{\Lambda}(P^{\Lambda}(s)) = P^{\Lambda}(0) \,  \quad \mbox{whenever } |s| \leq s_0^{\Lambda}(\gamma) \, ,
\end{equation}
where by $P^{\Lambda}(s)$ we denote the orthogonal projection onto the ground state space of $H(\Lambda, s)$. 

For any $\Lambda \in \mathcal{P}_0( \Gamma)$ and $s \in \mathbb{R}$, the
state $\omega_s^{\Lambda}$ given by
\begin{equation} \label{pert_state}
\omega_s^{\Lambda}(A) = \omega_0( \alpha_s^{\Lambda}(A)) \quad \mbox{for all } A \in \cA \, 
\end{equation}
is a vector state in the GNS space:
%In fact, using (\ref{gs_as_vs}) and (\ref{rep_spec_flow}), it is clear that 
\begin{equation} \label{pert_vec_state}
\omega_s^{\Lambda}(A) = \langle \Omega(\Lambda, s), \pi_0(A) \Omega (\Lambda, s) \rangle \quad \mbox{for all } A \in \cA \, ,
\end{equation} 
where $\Omega(\Lambda, s) = \tilde{U}(\Lambda,s) \Omega \in \cH$. By our assumptions, $P^{\Lambda}(0) = \ketbra{\Omega}$.
An application of (\ref{gen_spec_flow_proj_bh}) then shows that 
\begin{equation}
P^{\Lambda}(s) = \tilde{U}(\Lambda,s) \ketbra{\Omega} \tilde{U}(\Lambda,s)^* = \ketbra{\Omega(\Lambda,s)} \quad \mbox{for all } |s| \leq s_0^{\Lambda}(\gamma),
\end{equation}
and thus $\Omega(\Lambda,s)$ is the ground state of $H(\Lambda, s)$.

Finally, we recall that with the parameters $\gamma$ and $s$ as above that the weighted integral operator $\tilde{\cF}_s^{\Lambda}$ from (\ref{wio_in_GNS}) satisfies the relation
\begin{equation} \label{wi_commute_property}
\left[\tilde{\cF}^\Lambda_{s}(A), \ketbra{\Omega(\Lambda,s)}\right] = 0 \quad \mbox{for all } A \in \cB(\cH) \,.
\end{equation}
See, e.g. \cite[Lemma 6.8]{nachtergaele:2019}, for a proof of this property.

\subsection{On Domains} \label{sec:domain}

Recall that $(\Gamma, d)$ is a $\nu$-regular metric space.
Let $F$ be an $F$-function on $(\Gamma, d)$, and $\Phi$ an 
interaction with $\Vert \Phi\Vert_F<\infty$. As in Section~\ref{sec:main}, let $\delta^{\Phi}$ be the closed derivation
with dense domain ${\rm dom}( \delta^{\Phi}) \subset \cA$, and which satisfies
\be\label{derivation}
\delta^\Phi(A) = \sum_{\stackrel{Y \in \mathcal{P}_0( \Gamma):}{Y \cap X \neq \emptyset}} [ \Phi(Y) , A] \quad \mbox{for any } A \in \cA_X \, .
\ee
Although the sum on the right-hand-side above may be infinite, it is absolutely convergent when $\Phi$ 
has a finite $F$-norm. In fact, $\delta^{\Phi}$ is locally bounded: 
\be\label{Dvolume}
\Vert \delta^\Phi (A) \Vert \leq 2\Vert F\Vert \Vert \Phi\Vert_F |X| \Vert A\Vert  \quad \mbox{for all } X\in\cP_0(\Gamma), A\in \cA_X,
\ee
see Example 4.7 of \cite[Section IV.B.1]{nachtergaele:2019}. 
%By construction, it is clear that $\cA^{\rm loc} \subset {\rm dom}( \delta^{\Phi})$;
%in fact, $\cA^{\rm loc}$ is a core for each such $\delta^{\Phi}$. 
We have the following lemma.

\begin{lemma}\label{lem:domains}
Let $(\Gamma, d)$ be $\nu$-regular, $F$ an $F$-function on $(\Gamma, d)$, and $g$ a decay function with a finite $\nu$-moment, i.e., 
\begin{equation}
\sum_{n=1}^{\infty} (n+1)^{\nu} g(n) < \infty \, .
\end{equation}
For any interaction $\Phi$ on $\Gamma$ with $\| \Phi \|_F< \infty$, 
we have that $\cA^g\subset \dom (\delta^\Phi)$.
\end{lemma}
\begin{proof}
For $n\geq 1$,  and $A\in\cA^g$, for some $x\in\Gamma$, and observables 
$A_n \in \cA_{b_x(n)}$ satisfying $\Vert A - A_n\Vert\leq  \Vert A \Vert_{x, g} g(n)$. 
In this case, the bound $\Vert A_{n+1} - A_n\Vert \leq 2  \Vert A \Vert_{x, g} g(n)$ is clear. Using \eq{Dvolume} and $\nu$-regularity of $\Gamma$, we conclude
\be
\Vert \delta^\Phi(A_{n+1}) -\delta^\Phi(A_n) \Vert \leq 4\Vert F\Vert \Vert \Phi\Vert_F \kappa  \Vert A \Vert_{x, g} (n+1)^\nu g(n).
\ee
Thus, for all $m<n$,
\be \label{basic_der_est}
\Vert \delta^\Phi(A_{n}) -\delta^\Phi(A_m) \Vert \leq 4\kappa  \Vert A \Vert_{x, g} \Vert F\Vert \Vert \Phi\Vert_F \sum_{k=m}^{n-1} (k+1)^\nu g(k).
\ee
Since we assumed that $g$ has a finite $\nu$-moment, this implies that $\delta^{\Phi}(A_n) $ is a
Cauchy sequence. Since  $A_n\to A$ and $\cA^{\rm loc}$ is a core for $\delta^\Phi$, it follows
that $A \in \dom (\delta^\Phi)$.
\end{proof}

%Let us briefly comment on a basic application of this lemma. 
Given the assumptions of Theorem \ref{thm:finite_volume},
Lemma~\ref{lem:domains} clearly applies to the derivation $\delta_0$. 
Using that $H_0 \Omega = 0$, one readily checks the relation
\be
e^{itH_0} \pi_0(A) \Omega = \pi_0(\tau_t^{(0)} (A) ) \Omega \, ,
\ee
{f}rom which the inclusion $\pi_0(\dom (\delta_0))\Omega \subset \dom(H_0)$ is clear. 
As a result, if $g$ is a decay function with a finite $\nu$-moment, then 
$\pi_0 (\cA^g) \Omega \subset \dom (H_0) = \dom(H(\Lambda, s))$ for any
$\Lambda \in \mathcal{P}_0( \Gamma)$ and $s \in \mathbb{R}$. 
Since $U(\Lambda,s)\in\cA^g$ for some $g$ with finite moments of all of orders by \eq{spec_flow_uni_ql_est}, it follows that 
$\pi_0(A U(\Lambda,s)) \Omega\in\dom(H(\Lambda,s))$, for any $A\in\cA^{\rm loc}$, $s\in\Rl$.

A consequence of this is a gap inequality for the perturbed ground state $\omega_s^{\Lambda}$ from \eq{pert_state}.  Namely, we show that 
for $\gamma\in (0,\gamma_0), \Lambda\in\cP_0(\Gamma)$, and $|s| \leq s_0^{\Lambda}(\gamma)$:
\begin{equation}\label{gap_ineq}
	\omega_s^{\Lambda}(A^* \delta_s^{\Lambda}(A)) \geq \gamma \omega_s^{\Lambda}(A^*A) , \mbox{ for all } A \in \cA^{\rm loc} \mbox { with }\omega_s^{\Lambda}(A) =0.
\end{equation}
To see this, fix $|s| \leq s_0^{\Lambda}(\gamma)$. Since $\Omega(\Lambda,s)$ is the unique ground state of $H(\Lambda,s)$, 
\begin{equation} \label{excited_states}
\langle \psi, (H(\Lambda, s) - E(\Lambda, s) \idty) \psi \rangle \geq \gamma \langle \psi, \psi \rangle
\end{equation}
for all $\psi \in {\rm dom}(H_0)$ with $\langle \Omega(\Lambda, s), \psi \rangle =0$. 
In particular, if $\omega^\Lambda_s(A)=0$ for some $A\in\cA^{\rm loc}$, then \eq{excited_states} holds for
$\psi = \pi_0(A) \Omega(\Lambda, s)=\pi_0(A U(\Lambda,s)) \Omega$ since
$$
\langle \Omega(\Lambda,s),\psi\rangle = \langle \Omega(\Lambda,s), \pi_0(A) \Omega(\Lambda,s)\rangle=\omega_s^\Lambda(A).
$$
Then \eq{gap_ineq} follows from rewriting \eq{excited_states}.

It will be important that on an appropriate dense domain, the action of the unbounded Hamiltonians can be expressed as a 
limit of finite-volume quantities. This is the content of the next lemma.

\begin{lemma}\label{lem:strong_convergence_of_sums}
Let  $(\cH,\pi_0,\Omega)$ be the GNS representation of $\omega_0$, an infinite-volume, zero energy, ground state of a frustration free model 
as in Assumption \ref{ass:initial_model}. For any decay function $g$ with a finite $\nu$-moment  and any IAS $(\Lambda_n)$,
\be \label{lim_sum_ham_psi}
\lim_{n \to \infty} \pi_0(H_{\Lambda_n}) \psi = H_0\psi \quad \mbox{ for all } \psi\in \pi_0(\cA^g) \Omega \, ,
\ee
where $H_{\Lambda_n} \in \cA_{\Lambda_n}$ is as in (\ref{HLambdan}) and $H_0$ is the GNS Hamiltonian.
\end{lemma}

\begin{proof}%[Proof of Lemma~\ref{lem:strong_convergence_of_sums}:]
Note that \eq{lim_sum_ham_psi} is trivially satisfied for $\psi=\pi_0(A)\Omega$, for $A\in\cA^{\rm loc}$ since
\be\label{locals}
\lim_{n\to\infty}\pi_0(H_{\Lambda_n})\pi_0(A)\Omega= \lim_{n\to\infty}\pi_0([H_{\Lambda_n},A])\Omega = \pi_0(\delta_0(A))\Omega=H_0\pi_0(A)\Omega.
\ee
For the first equality we used  $\pi_0(H_{\Lambda_n})\Omega=0$, which is a consequence of the frustration-free property. Then, by the finite-range 
condition on the unperturbed model, $[H_{\Lambda_n},A]$ becomes constant for $n$ sufficiently large.

Take $\psi = \pi_0(A) \Omega$ for any $A \in \cA^g$. By the definition of $\cA^g$, there exists $x \in \Gamma$ and observables $A_m \in \cA_{b_x(m)}$
so that $\| A - A_m\| \leq \Vert A \Vert_{x, g} g(m)$ for all $m \geq 1$, and so the vectors $\psi_m := \pi_0(A_m) \Omega$ satisfy
\begin{equation}\label{convergence}
\| \psi - \psi_m \| \leq  \Vert A \Vert_{x, g} g(m). %\quad \mbox{and} \quad \| \psi_m - \psi_k \| \leq 2 C g(k) \quad \mbox{whenever } 1 \leq k < m \, .  
\end{equation}

Moreover, since the interaction $h$ is uniformly bounded with range $R$, it follows from \eqref{derivation} and $\nu$-regularity that for any $k\geq 1$,
\[
\|\delta_0(A_{\ell+1}-A_\ell)\| \leq  \sum_{x\in b_{x}(\ell+1+R)}\|[h_x,A_{\ell+1}-A_\ell]\| \leq 2\kappa^2R^\nu(\ell+1)^\nu\|h\|_\infty\|A_{\ell+1}-A_\ell\|
\]
where we use the bound $|b_{x}(n+m)|\leq |b_x(n)||b_x(m)|$. Then, by the last equality of \eq{locals}, one finds that $(H_0\psi_m)_{m\in\bN}$ is Cauchy as
\begin{equation}\label{cauchy}
\| H_0 \psi_m - H_0 \psi_k \| \leq \| \delta_0(A_m) - \delta_0(A_k) \| \leq D \sum_{\ell = k}^{m-1} (\ell+1)^{\nu} g( \ell) \, .
\end{equation}
where we set $D = 4 \kappa^2 R^{\nu} \| h \|_{\infty} \Vert A \Vert_{x, g}$. Since $H_0$ is closed, and $\psi\in\dom(H_0)$ by Lemma~\ref{lem:domains} and the subsequent discussion, the bound
\begin{equation} \label{H_0_est}
\| H_0 \psi - H_0 \psi_k \| \leq  D \sum_{\ell = k}^{\infty} (\ell+1)^{\nu} g( \ell) \, 
\end{equation}
follows immediately from \eqref{convergence}-\eqref{cauchy}.

In the case of a local Hamiltonian, using again the first equality in \eqref{locals}, a similar argument shows that for all $n \geq 1$, 
\begin{align} 
	\| \pi_0(H_{\Lambda_n}) \psi - \pi_0(H_{\Lambda_n}) \psi_k \| & =\lim_{m\to \infty}  \| \pi_0(H_{\Lambda_n}) \psi_m - \pi_0(H_{\Lambda_n}) \psi_k \| \nonumber \\
	&  \leq \lim_{m\to \infty}  \| [H_{\Lambda_n}, A_m- A_k] \| \leq D \sum_{\ell = k}^{\infty} (\ell+1)^{\nu} g( \ell) \,.  \label{H_n_est}
\end{align}
%\begin{equation}
%\| \pi_0(H_{\Lambda_n}) \psi_m - \pi_0(H_{\Lambda_n}) \psi_k \| \leq \| [H_{\Lambda_n}, A_m- A_k] \| \leq 2 \kappa C D_0 \sum_{\ell = k}^{m-1} (\ell+1)^{\nu} g( \ell) 
%\end{equation}
%for all $1 \leq k < m$ and therefore, we obtain
%\begin{equation} \label{H_n_est}
%\| \pi_0(H_{\Lambda_n}) \psi - \pi_0(H_{\Lambda_n}) \psi_k \| \leq 2 \kappa C D_0 \sum_{\ell = k}^{\infty} (\ell+1)^{\nu} g( \ell) \, 
%\end{equation}
Putting all of this together, one finds that for any $n \geq 1$ and each $k \geq 1$,
\begin{align} \label{3_term_bd}
\| \pi_0(H_{\Lambda_n}) \psi - H_0 \psi \| & \leq \| \pi_0(H_{\Lambda_n}) \psi - \pi_0(H_{\Lambda_n}) \psi_k  \| + \| \pi_0(H_{\Lambda_n}) \psi_k - H_0 \psi_k \|   + \| H_0 \psi_k - H_0 \psi \| 
\end{align}
For $k \geq 1$ sufficiently large, (\ref{H_0_est}) and (\ref{H_n_est}) guarantee that the first and last term above can be made arbitrarily small.
Given such a $k$, the middle term vanishes for $n$ sufficiently large, see the comment following \eq{locals}. This completes the proof.
\end{proof}

Lemma~\ref{lem:strong_convergence_of_sums} also trivially applies to the perturbed system
in the GNS space. In fact, for  $\Lambda \in \mathcal{P}_0(\Gamma)$ and $s \in \mathbb{R}$, under assumptions as above, 
 a direct application of Lemma~\ref{lem:strong_convergence_of_sums} shows that we also have
\be
\lim_n  \pi_0(H_{\Lambda_n} + s V_{\Lambda})\psi = H_0 \psi + s\pi_0(V_\Lambda)\psi = H(\Lambda, s) \psi \quad \mbox{for all } \psi\in \pi_0(\cA^g)\Omega.
\ee

\begin{remk} An analogue of Lemma~\ref{lem:strong_convergence_of_sums} holds more generally. In fact, if 
$F$ is an $F$-function  with a finite $\nu$-moment,
then for any frustration free interaction $\Phi$ %, see \cite[Section 2.2.1]{},  
with $\| \Phi \|_F< \infty$, the GNS Hamiltonian again satisfies (\ref{lim_sum_ham_psi}). The argument is identical to the
above except that one uses the more general estimate in Lemma~\ref{lem:domains} and bounds the middle
term in (\ref{3_term_bd}) by 
\begin{equation}
\sum_{\stackrel{X \in \mathcal{P}_0( \Gamma):}{X \cap b_x(k) \neq \emptyset, X \cap \Lambda_n^c \neq \emptyset}} \| \pi_0(\Phi(X)) \psi_k \| \leq \| \Phi \|_F (1+g(0)) \| A \|_{x,g} \sum_{y \in b_x(k), z \in \Lambda_n^c} F(d(y,z))
\end{equation}
For fixed $k$, the above is the sum of finitely many `tails' of the uniformly summable function $F$. 
\end{remk}

We now investigate how the weighted integral operator $\tilde{\cF}_s^{\Lambda}$ from
(\ref{wio_in_GNS}) can be applied to the unbounded Hamiltonian $H(\Lambda, s)$. 
To begin, we prove an analogue of the desired statement for the unperturbed dynamics; this is 
Lemma~\ref{lem:FH_convergence} below. To this end, assume $w: \mathbb{R} \to \mathbb{R}$ satisfies
\begin{equation} \label{gen_w_constraints}
\int_{\mathbb{R}} w(u) \, du = 1 \quad \mbox{and} \quad \int_{\mathbb{R}} |u|^{\nu} |w(u)| \, du < \infty \, ,
\end{equation}
and define a weighted integral operator $\tilde{\cF}: \mathcal{B}( \cH) \to \mathcal{B}(\cH)$ by setting
\be \label{rep_gen_F_int_op}
\tilde{\cF}(A) = \int_{-\infty}^\infty  e^{iuH_0}A e^{-iuH_0} w(u) \, du \quad \mbox{for all } A\in \cB(\cH) \, .
\ee
To simplify notation, let us also write
\begin{equation} \label{rep_unpert_dyn_3}
\tilde{\tau}^{(0)}_u(A) = e^{iuH_0}A e^{-iuH_0}. %\quad \mbox{for all } A\in \cB(\cH) \mbox{ and } u \in \mathbb{R} \, .
\end{equation}
Our first result is as follows.
\begin{lemma}\label{lem:FH_convergence}
Let $(\Gamma,d)$ be $\nu$-regular, $g$ be a decay function with a finite $2 \nu$-moment, and $w: \mathbb{R} \to \mathbb{R}$ satisfy
(\ref{gen_w_constraints}). For each choice of IAS $(\Lambda_n)$,
the weighted integral operator $\tilde{\cF}$ from (\ref{rep_gen_F_int_op}) satisfies
\be \label{FH_convergence}
\lim_{n \to \infty} \tilde{\cF}(\pi_0(H_{\Lambda_n})) \psi = H_0\psi \quad \mbox{ for all } \psi\in\pi_0(\cA^g)\Omega \, ,
\ee
where $H_{\Lambda_n} \in \cA_{\Lambda_n}$ is as in (\ref{HLambdan}).
\end{lemma}
\begin{proof}
Fix an IAS $(\Lambda_n)$ and take $\psi = \pi_0(A) \Omega$ for some $A \in \cA^g$. 
We can rewrite the convergence claimed in (\ref{FH_convergence}) as the convergence of integrals 
of a sequence of functions $f_n : \mathbb{R} \to \cH$ given by  
\be
f_n(u) = w(u) \tilde \tau_u^{(0) }( \pi_0(H_{\Lambda_n} ) ) \psi \quad \mbox{for all } u \in \mathbb{R} \, .
\ee
Since $H_0\Omega =0$, the above can be re-written as
\be\label{fnu}
f_n(u) = w(u) e^{iuH_0} \pi_0(H_{\Lambda_n}) \pi_0 (\tau_{-u}^{(0)} (A))\Omega
\ee
using (\ref{rep_unpert_dyn}). We claim that there is a decay function $g_{\tau}$ with a finite $\nu$-moment
such that $\tau_{-u}^{(0)}(A) \in \cA^{g_{\tau}}$ for all $u \in \mathbb{R}$. Given this, Lemma~\ref{lem:strong_convergence_of_sums} applies
and we find that
\begin{equation}
\lim_{n \to \infty} f_n(u) = w(u) e^{iuH_0} H_0 \pi_0(\tau_{-u}^{(0)}(A)) \Omega = w(u) H_0 \psi \quad \mbox{for all } u \in \mathbb{R} \, .
\end{equation}
By (\ref{gen_w_constraints}), the integral of this limit coincides with the right-hand-side of
(\ref{FH_convergence}). Therefore, to complete the proof we only need to justify 
an application of dominated convergence.

Let us first prove the existence of a decay function $g_{\tau}$ as claimed. Fix $A \in \cA^g$. In this case, there is $x \in \Gamma$, $C \geq 0$, 
and observables $A_m \in \cA_{b_x(m)}$ for which $\| A - A_m \| \leq C g(m)$ for all $m \in \mathbb{N}$.
Let $u \in \mathbb{R}$ and for any $n \in \mathbb{N}$, set 
\begin{equation}
A_n(u) = \Pi_n (\tau_u^{(0)} (A_{\lceil n/2 \rceil}))\in\cA_{b_x(n)} 
\end{equation}
where, to ease notation, we have written $\Pi_n = \Pi_{b_x(n)}$,  for the conditional expectation from Section~\ref{sec:qlm_gf}. 
A straightforward estimate shows that for any $\mu>0$,
\begin{eqnarray}
\Vert \tau_u^{(0)} (A) - A_n(u) \Vert & \leq & \Vert \tau_u^{(0)} (A) - \tau_u^{(0)} (A_{\lceil n/2 \rceil})\Vert + \Vert \tau_u^{(0)} (A_{\lceil n/2 \rceil})- A_n(u)) \Vert \nonumber \\
& \leq & C g( n/2) + \kappa C_{\mu}e^{\mu (v_{\mu}|u| +1)} (1+g(0)) \Vert A\Vert_{x,g}(n/2 +1)^{\nu}  e^{- \mu n/2}
\end{eqnarray} 
where we used (\ref{unpert_lrb}) and (\ref{cond_exp_est}) for the final bound. The existence of the  decay function $g_{\tau}$
follows from the moment condition on $g$ and the decay of the exponential term.. 

We now turn to finding a dominating function for $f_n$. Recall that  for any $m_0 \in \mathbb{N}$, 
$A$ can be written as an absolutely convergent, telescopic sum:
\be\label{def:Bn}
A = A_{m_0} + \sum_{k=m_0+1}^\infty B_k \quad \mbox{where} \quad B_k= A_k-A_{k-1} \quad \mbox{ and } \quad \Vert B_k\Vert \leq 2Cg(k-1).
\ee
Inserting this decomposition of $A$ into (\ref{fnu}), we find that for any $n \in \mathbb{N}$ and each $u \in \mathbb{R}$:
\be\label{fn_bd_1}
\Vert f_n(u)\Vert \leq  |w(u)|  \left( \Vert \pi_0(H_{\Lambda_n}) \pi_0(\tau_{-u}^{(0)} (A_{m_0}))\Omega\Vert
+ \sum_{k=m_0+1}^\infty \Vert \pi_0(H_{\Lambda_n}) \pi_0(\tau_{-u}^{(0)} (B_k))\Omega\Vert \right) \, .
\ee

Now, by the zero-energy property of the ground state
%, it is clear that
%\begin{equation}
%\pi_0(H_{\Lambda_n}) \pi_0(A) \Omega = \pi_0([H_{\Lambda_n},A]) \Omega \quad \mbox{for all } A \in \cA^{\rm loc},
%end{equation} 
%and thus, in particular, 
we find the  bound
\begin{equation} \label{uni_bd_1}
\| \pi_0(H_{\Lambda_n}) \pi_0(A) \Omega \| \leq \kappa (k+R)^{\nu} \Vert h\Vert_\infty \| A \| \quad \mbox{for all } A \in \cA_{b_x(k)} \, ,
\end{equation}
which we stress is uniform in $n$. This suggests a mechanism for bounding the first term in (\ref{fn_bd_1}).
Let $\ell_0 \geq m_0$ and write %conditional expectations, as in (\ref{delta_ops}), to write
\begin{equation} \label{2nd_tele}
\tau_{-u}^{(0)}(A_{m_0}) = \sum_{\ell \geq \ell_0} \Delta_{\ell_0}^{\ell}(\tau_{-u}^{(0)}(A_{m_0})) \, ,
\end{equation}
where we have used the short-hand $\Delta_m^{\ell}$ for $\Delta_{b_x(m)}^{\ell}$ as in (\ref{delta_ops}).
For $\ell = \ell_0$, the bound
\begin{eqnarray}
 \| \pi_0(H_{\Lambda_n}) \pi_0(\Delta_{\ell_0}^{\ell_0}(\tau_{-u}^{(0)}(A_{m_0})) ) \Omega \| & = & \| \pi_0(H_{\Lambda_n}) \pi_0(\Pi_{\ell_0}(\tau_{-u}^{(0)}(A_{m_0})) ) \Omega \| \nonumber \\
 & \leq & \kappa \Vert h\Vert_\infty (\ell_0 +R)^{\nu} \| A_{m_0} \| 
\end{eqnarray}
follows from (\ref{uni_bd_1}). For $\ell \geq \ell_0$, the estimate
\begin{eqnarray}
 \| \pi_0(H_{\Lambda_n}) \pi_0(\Delta_{\ell_0}^{\ell}(\tau_{-u}^{(0)}(A_{m_0})) ) \Omega \| & \leq & \kappa \Vert h\Vert_\infty (\ell +R)^{\nu} \| \Delta_{\ell_0}^{\ell}(\tau_{-u}^{(0)}(A_{m_0})) \| \nonumber \\
 & \leq & 2\kappa^2  \Vert h\Vert_\infty C_{\mu} m_0^{\nu} (\ell + R)^{\nu}  \| A_{m_0} \| e^{- \mu( \ell - m_0 -1- v_{\mu} |u|)} 
\end{eqnarray}
follows from another application of (\ref{uni_bd_1}) and the quasi-locality estimate for the unperturbed dynamics in combination with (\ref{ce_diff_bd}). 
We conclude that 
\begin{align}\label{hard_est}
\| \pi_0(H_{\Lambda_n}) \pi_0(\tau_{-u}^{(0)} (A_{m_0}))\Omega \| & \leq  \sum_{\ell \geq \ell_0} \| \pi_0(H_{\Lambda_n}) \pi_0(\Delta_{\ell_0}^{\ell}(\tau_{-u}^{(0)}(A_{m_0})) ) \Omega \| \nonumber \\
& \leq  \kappa \Vert h\Vert_\infty \| A_{m_0} \| \left( (\ell_0 +R)^{\nu} + 2 \kappa C_{\mu} m_0^{\nu} \sum_{\ell \geq \ell_0+1} (\ell + R)^{\nu} e^{-  \mu(\ell - m_0 -1- v_{\mu} |u|)}  \right) .
\end{align}
If we now take $\ell_0 = \lceil v_{\mu}|u| + m_0 \rceil$, then we have found that there is $K \geq 0$ for which 
\begin{equation} \label{unpert_lem_final_bd}
\| \pi_0(H_{\Lambda_n}) \pi_0(\tau_{-u}^{(0)} (A_{m_0}))\Omega \| \leq K\Vert h\Vert_\infty  \| A_{m_0} \| \left( m_0^{2 \nu} + |u|^{\nu} + (m_0 |u|)^{\nu} + 1 \right) 
\end{equation}
and here $K = K(\kappa, \mu, \nu, R)$. 

The terms $B_k$ in \eq{fn_bd_1} can be estimated similarly.
Regarding $k$ as $m_0$ and arguing as in (\ref{2nd_tele}) - (\ref{hard_est}) with some $\ell_0 \geq k$,
a bound analogous to (\ref{unpert_lem_final_bd}) can be found. Of course, here one replaces $\| A_{m_0} \|$ with $\| B_k \|$. 
Since $\| B_k \| \leq 2 C g(k-1)$ and $g$ has a finite $2 \nu$-moment, we have obtained a bound on the right-hand-side of
(\ref{fn_bd_1}) of the form:
\begin{equation}
\| f_n(u) \| \leq \tilde{K}\Vert h\Vert_\infty  (1+ |u|^{\nu}) |w(u)| \quad \mbox{for all } u \in \mathbb{R} \, .
\end{equation}
By the assumption on $w$, i.e. (\ref{gen_w_constraints}), the above is a dominating function for the sequence $f_n$.
This justifies dominated convergence and completes the proof. 
\end{proof}

We will also need a version of Lemma \ref{lem:FH_convergence} for the perturbed system. 
Recall that for any $\gamma >0$, $s \in \mathbb{R}$, and
$\Lambda \in \mathcal{P}_0(\Gamma)$,  the weighted integral operator $\tilde{\cF}_s^{\Lambda} : \cB(\cH) \to \cB(\cH)$ are defined by
\begin{equation} \label{rep_pert_F_int_op}
	\tilde{\cF}_s^{\Lambda}(A) = \int_{- \infty}^{\infty} e^{itH(\Lambda, s)} A e^{-itH(\Lambda, s)} w_{\gamma}(t) \, dt \quad \mbox{for all } A \in \cB(\cH) \, .
\end{equation} 
We note that $w_\gamma$ from \cite[Section VI.B]{nachtergaele:2019} satisfies (\ref{gen_w_constraints}). It is clear that
\be\label{wi_Ham_invar}
\tilde{\caF}^\Lambda_{s}(e^{iuH(\Lambda,s)}) = e^{iu H(\Lambda,s)} \quad \mbox{for all } \;  u,s \in\bR \, ,
\ee
since the dynamics leaves this bounded operator invariant and $w_\gamma$ integrates to 1.
Lemma \ref{lem:FsHs_convergence} provides a  differential version of this fact. 

\begin{lemma}\label{lem:FsHs_convergence}
Let $(\Gamma,d)$ be $\nu$-regular, $g$ be a decay function with a finite $2 \nu$-moment. 
Let $\Lambda \in \mathcal{P}_0( \Gamma)$ and take $s \in \mathbb{R}$. For each choice of IAS $(\Lambda_n)$, consider
the weighted integral operator $\tilde{\cF}_s^{\Lambda}$, as in (\ref{rep_pert_F_int_op}), with arbitrary  $w: \mathbb{R} \to \mathbb{R}$ satisfying
(\ref{gen_w_constraints}). Then
\be \label{FsHs_convergence}
\lim_{n \to \infty} \tilde{\cF}_s^{\Lambda}(\pi_0(H_{\Lambda_n} +s V_{\Lambda})) \psi = H(\Lambda, s)\psi \quad \mbox{ for all } \psi\in\pi_0(\cA^g)\Omega \, ,
\ee
with $H_{\Lambda_n} \in \cA_{\Lambda_n}$ as in (\ref{HLambdan}) and $V_{\Lambda}$ as in (\ref{perturbation}).
\end{lemma}
\begin{proof}
Fix an IAS $(\Lambda_n)$ where we assume for convenience that $\Lambda \subset \Lambda_1$.
As in the proof of Lemma \ref{lem:FH_convergence}, take $\psi = \pi_0(A) \Omega$ with $A \in \cA^g$, and for each $n \in \mathbb{N}$,
consider $f_n: \mathbb{R} \to \cH$ given by
\be
f_n(u) = w(u) \tilde \tau_u^{(\Lambda, s)}( \pi_0(H_{\Lambda_n}+s V_\Lambda) ) \pi_0(A) \Omega \quad \mbox{for all } u \in \mathbb{R} \, ,
\ee
where, in analogy to (\ref{rep_unpert_dyn_3}), we have set 
\begin{equation} \label{rep_pert_dyn_3}
\tilde{\tau}_t^{(\Lambda, s)}(A) = e^{itH(\Lambda, s)} A e^{-itH(\Lambda,s)} \quad \mbox{for all } A \in \cB(\cH) \mbox{ and } t \in \mathbb{R} \, .
\end{equation}
Using (\ref{equal_pert_dyn}), (\ref{int_pic_dyn}), and (\ref{rep_unpert_dyn}), we may write 
\begin{equation}
\tilde{\tau}_u^{(\Lambda, s)}( \pi_0(H_{\Lambda_n} + s V_\Lambda)) = \pi_0( K_u^{(\Lambda, s)})^* e^{iuH_0} \pi_0(H_{\Lambda_n} + s V_{\Lambda}) e^{-iuH_0} \pi_0(K_u^{(\Lambda, s)} )
\end{equation}
for all $u \in \mathbb{R}$. In this case, we find that
\be
f_n(u) = w(u) \pi_0(K_u^{(\Lambda,s)})^* e^{iuH_0} \pi_0(H_{\Lambda_n} + s V_\Lambda) \pi_0( \tau^{(0)}_{-u}( K_u^{(\Lambda,s)} A))\Omega.
\ee
Following a similar argument and using (\ref{int_pic_uni_ql_est}), one shows that there is a decay function $g'$ with a finite $\nu$-moment such that
$\tau^{(0)}_{-u}( K_u^{(\Lambda,s)} A) \in \cA^{g'}$. As a result, the point-wise limit
\be
\lim_n f_n(u) = w(u) \pi_0(K_u^{(\Lambda,s)})^* e^{iuH_0} (H_0 + s \pi_0(V_\Lambda)) \pi_0( \tau^{(0)}_{-u}( K_u^{(\Lambda,s)} A))\Omega 
= w(u) H(\Lambda,s) \psi
\ee
is clear from properties of the interaction picture dynamics, see the discussion following (\ref{int_pic_dyn}). 

The argument demonstrating that we can apply the dominated convergence theorem also proceeds as in the proof
of Lemma~\ref{lem:FH_convergence}. Since the differences stemming from the presence
of the $u$-dependence in the operators $A_{m_0}$ and $B_k$ are minor, we leave the details to the reader.
\end{proof}

\subsection{Proof of bounded differentiability of the spectral flow}\label{sec:bounded_differentiability}

At this point we pause to prove a result in the setting of this paper that is useful for related considerations concerning gapped insulators and symmetry protected phases.
For this purpose we consider differentiable curves of interactions $[0,1]\ni s \mapsto \Phi(s)$, for which we assume a finite norm of the following form:
\begin{equation}
	\| \Phi \|^1_F=\sup_{s\in [0,1]} \sup_{x,y \in \Gamma} \frac{1}{F(d(x,y)) } \sum_{\stackrel{X \in \mathcal{P}_0(\Gamma):}{x,y \in X}} \left( \| \Phi(X,s) \| + |X| \| \Phi'(X,s) \| \right),
\end{equation}
where for each finite $X \subset \Gamma$ and $0 \leq s \leq 1$, $\Phi(X,s)^* = \Phi(X, s) \in \mathcal{A}_X$, $\Phi(X,\cdot): [0,1] \to \mathcal{A}_X$ is 
differentiable, and $F$ is an $F$-function of stretched exponential decay as in \eq{class_of_Fs}.
%there is $a>0$, $\theta \in (0,1]$, and $\xi > \nu +1$ for which 
%\begin{equation}
%	F(r) = \frac{e^{-ar^{\theta}}}{(1+r)^{\xi}} \quad \mbox{for all } r \geq 0 \, .
%\end{equation}
For any differentiable curves of interactions $\Phi$ with finite norm  $\| \Phi \|^1_F$, there is an $F$-function $\tilde{F}$, and an $s$-dependent interaction
$\Psi(s)$ with 
$$
\| \Psi \|_{\tilde{F}} := \sup_{s\in[0,1]} \Vert \Psi(s)\Vert_{\tilde F} < \infty,
$$
that generates the  infinite volume spectral flow automorphisms $\alpha_s$ \cite[Section VI.E.2]{nachtergaele:2019}.  
We further know by \cite[Theorem 3.9]{nachtergaele:2019}, that
this infinite volume spectral flow is differentiable on $\mathcal{A}_{\Gamma}^{\rm loc}$ with 
\begin{equation}
	\frac{d}{ds} \alpha_s(A) = i \sum_{X \in \mathcal{P}_0( \Gamma)} \alpha_s([\Psi(X,s), A]) \quad \mbox{for all } A \in \mathcal{A}_{\Gamma}^{\rm loc} \, .
\end{equation} 
Our aim is the show that 
$$
\alpha'_s :=\frac{d}{ds} \alpha_s 
$$
is a bounded map when defined on a suitable Banach algebra  of $g$-local obervables. 
%
%Such a Banach space is  defined by a norm $\Vert\cdot\Vert_g$ for a decay function $g$, and increasing and absorbing sequence of finite subsets of $\Gamma$, $(\Lambda_n)_{n\geq 1}$, as follows:
%\begin{equation} \label{f-norm}
%	\| A \|_g = \| A \| + \sup_{n \in \mathbb{N}} \left( \frac{\| A - \Pi_{n}(A) \| }{ g(n)} \right),\quad  \cA^{(\Lambda_n),g} =\overline{\cA^{\rm loc}}^{\Vert \cdot\Vert_g}.
%\end{equation}
%For a proof that $\cA_g$ is the Banach space of all $A\in\cA_\Gamma$ for which $\Vert A\Vert_g<\infty$, see \cite{moon:2020}. Note that for the choice $\Lambda_n =b_x(n)$ and
%$g$ a stretched exponetial of the type \eq{class_of_Fs}, $ \cA^{(\Lambda_n),g} $ coincides with $\cA^g$ of Section \ref{sec:setup}.
%
%

\begin{thm}\label{thm:bounded_differentiability}
	Suppose $g$ is a decay function and $(\Lambda_n)$ an IAS that satisfy 
	\be\label{g_summable}
	\sum_{n\geq 1} |\Lambda_{n+1}| g(n) < \infty.
	\ee  Then, the derivative $\alpha'_s$ of the spectral flow $\alpha_s$ is a well-defined bounded linear map $\cA^{(\Lambda_n),g} \to \cA$, satisfying
	\begin{equation}
		\| \alpha'_s(A) \| \leq 2 \| \Psi \|_{\tilde{F}} \| \tilde{F} \|_1  \left( | \Lambda_1| +2 \sum_{n=1}^{\infty} | \Lambda_{n+1}| g(n) \right) \| A \|_{(\Lambda_n),g} \, ,
	\end{equation}
with $ \| \tilde{F} \|_1 $ as defined in \eq{uni_sum}.
\end{thm}

\begin{proof}
	We start by noting that for any local observable $A \in \mathcal{A}_Y$ with $Y \subset \Gamma$ finite, we have the following estimate:
	\begin{eqnarray} \label{local_derivative}
		\left\|  \frac{d}{ds} \alpha_s(A) \right\| \leq \sum_{\stackrel{X:}{X \cap Y \neq \emptyset}} \| [ \Psi(X,s), \alpha_s (A)] \| & \leq & 2 \| A \| \sum_{y \in Y} \sum_{z \in \Gamma} 
		\sum_{\stackrel{X:}{y,z \in X }} \| \Psi(X,s) \| \nonumber \\
		& \leq & 2 \| \Psi \|_{\tilde{F}} \| A \| \sum_{y \in Y} \sum_{z \in \Gamma} \tilde{F}(d(y,z)) \nonumber \\
		& \leq & 2 \| \Psi \|_{\tilde{F}} \| \tilde{F} \|_1 |Y| \| A \|
	\end{eqnarray}
	
	Let $A \in \mathcal{A}^{(\Lambda_n),g}$. By (\ref{g-norm}), 
	\begin{equation} \label{f_conv}
		\| A - \Pi_n(A) \| \leq \| A \|_{(\Lambda_n),g}  g(n)
	\end{equation}
	which implies that the sequence $\{ \Pi_n(A) \}_{n=1}^{\infty}$ converges to $A$ (in norm) at a rate governed by $g$. 
	Define a map $\alpha'_s : \mathcal{A}^{(\Lambda_n),g}  \to \mathcal{A}$ by setting
	\begin{equation} \label{deriv_def_Af}
		\alpha'_s(A) = \lim_{n \to \infty}\frac{d}{ds}  \alpha_s \left( \Pi_n(A) \right) \quad \mbox{for all } A \in \mathcal{A}^{(\Lambda_n),g}  \, . 
	\end{equation}
	Note that for any strictly local observable $A$, $\Pi_n(A) = A$ if $n$ is sufficiently large, and thus this
	definition agrees with the standard definition of the derivative of $\alpha_s(A)$ for $A\in\cA^{\rm loc}$.
	
	Now, consider integers $1 \leq M < N < \infty$. For any observable $A$,
	\begin{equation}
		\Pi_N(A) - \Pi_M(A) = \sum_{n=M}^{N-1} \left( \Pi_{n+1}(A) - \Pi_n(A) \right)
	\end{equation}
	and for $A \in \mathcal{A}_g$, we also have that for any $n \geq 1$,
	\begin{eqnarray} \label{norm_bd}
		\| \Pi_{n+1}(A) - \Pi_n(A) \| & \leq & \| \Pi_{n+1}(A) - A \| + \| A - \Pi_n(A) \| \nonumber \\
		& \leq & \| A \|_{(\Lambda_n),g}  \left( g(n+1) + g(n) \right) \leq 2 \| A \|_{(\Lambda_n),g}  g(n)  
	\end{eqnarray} 
	In the above, we have used (\ref{f_conv}) and the fact that $g$ is non-increasing. 
	We conclude that 
	\begin{eqnarray}
		\| \alpha'_s( \Pi_N(A)) - \alpha'_s(\Pi_M(A)) \| & \leq & \sum_{n=M}^{N-1} \| \alpha'_s \left( \Pi_{n+1}(A) - \Pi_n(A) \right) \| \nonumber \\
		& \leq & 2 \| \Psi \|_{\tilde{F}} \| \tilde{F} \|_1 \sum_{n=M}^{N-1} |\Lambda_{n+1}| \| \Pi_{n+1}(A) - \Pi_n(A) \| \nonumber \\
		& \leq & 4 \| \Psi \|_{\tilde{F}} \| \tilde{F} \|_1 \| A \|_{(\Lambda_n),g}  \sum_{n=M}^{N-1} | \Lambda_{n+1}| g(n) \, .
	\end{eqnarray}
	Here, for the second inequality above we used (\ref{local_derivative}), and the final bound comes from (\ref{norm_bd}). 
	We conclude that whenever $g$ satisfies (\ref{g_summable}), the sequence $\{ \alpha'_s( \Pi_n(A)) \}_{n=1}^{\infty}$ is norm Cauchy; hence, 
	norm convergent. This shows that $\alpha'_s$ as given in (\ref{deriv_def_Af}) is well-defined and equals the derivative of $\alpha_s(A)$, for $A\in\cA^{(\Lambda_n),g} $. 
	Moreover, for any $N \geq 1$,
	\begin{eqnarray}
		\| \alpha'_s(\Pi_N(A)) \| & \leq & \| \alpha'_s(\Pi_1(A)) \| + \| \alpha'_s( \Pi_N(A)) - \alpha'_s(\Pi_1(A)) \|  \nonumber \\
		& \leq & 2 \| \Psi \|_{\tilde{F}} \| \tilde{F} \|_1 | \Lambda_1| \| \Pi_1(A) \| +  4 \| \Psi \|_{\tilde{F}} \| \tilde{F} \|_1 \| A \|_{(\Lambda_n),g}  \sum_{n=1}^{N-1} | \Lambda_{n+1}| g(n) 
	\end{eqnarray}
	and therefore,
	\begin{equation}
		\| \alpha'_s(A) \| = \lim_{N \to \infty} \| \alpha'_s( \Pi_N(A)) \| \leq 2 \| \Psi \|_{\tilde{F}} \| \tilde{F} \|_1  \left( | \Lambda_1| +2 \sum_{n=1}^{\infty} | \Lambda_{n+1}| g(n) \right) \| A \|_{(\Lambda_n),g} 
	\end{equation}
	as claimed.
\end{proof}

This proves that Assumption 1.2 (vii) in \cite{moon:2020} (see also \cite{bourne:2021,ogata:2021}) is always satisfied under the assumptions made in that reference. We note that in the same way one obtains the extension of the invariants constructed by Ogata in \cite{ogata:2021e,ogata:2020a,ogata:2021}  from models with finite-range interactions to the setting with interactions of stretched exponential decay.

\section{Construction of a unitarily equivalent perturbed system}\label{sec:finite-volume-perturbations}

%\subsection{Construction of a relatively bounded perturbation}

The crux of the stability strategy introduced in \cite{bravyi:2010}, is the construction of a unitarily equivalent perturbed system
using the spectral flow (aka quasi-adiabatic evolution) for which one can prove a relative form bound using quasi-locality estimates and LTQO.
In the infinite-system setting, this means
we need to prove that by transforming the unbounded Hamiltonian $H(\Lambda,s)$ from (\ref{fv_pert_ham} by the spectral flow
one arrives at an equivalent Hamiltonian of the form
\be\label{main_decomp}
\tilde{U}(\Lambda,s)^*H(\Lambda,s)\tilde{U}(\Lambda,s) \psi = H_0\psi + W(\Lambda,s) \psi + E(\Lambda,s)\psi \quad \mbox{for all } \psi \in \pi_0(\cA^{\rm loc})\Omega \, ,
\ee
with $W(\Lambda,s)$ a bounded operator with an explicit, $\Lambda$-independent form-bound with respect to $H_0$ and $E(\Lambda,s)$ the ground state
energy of $H(\Lambda, s)$ from (\ref{pert_gse}).  That $W(\Lambda,s)$ is well-defined is a consequence of the fact that  $\tilde U(\Lambda,s)\in\pi_0(\cA^g)$ for a function $g$ of the form \eq{class_of_Fs}, which guarantees that $\cA^g$ is an algebra. Hence $\tilde U(\Lambda,s)\pi_0(A)\Omega \in \dom H(\Lambda,s)$ for $A\in\cA^{\rm loc}$, by Lemma \ref{lem:strong_convergence_of_sums}. 

%
%In previous work, we established quasi-locality estimates for this spectral flow when the perturbation is of sufficiently rapid decay.
%Given Assumption \ref{ass:short-range}, the resulting quasi-locality can be used to establish that each $\psi \in \pi_0(\cA_{\rm loc})\Omega$ 
%is in the domain of the unitarily evolved Hamiltonian in \eqref{main_decomp} whenever $\Lambda \in \mathcal{P}_0( \Gamma)$ as $\tilde{U}(\Lambda,s)\psi \in \cA^g$ for an appropriate choice of $g$.

The proof of Theorem \ref{thm:finite_volume} is a consequence of two results. The first,Theorem \ref{thm:Step2}, establishes that $W(\Lambda,s)$ is indeed bounded and can be decomposed in way that is suitable for deriving a relative form bound. The second, Theorem \ref{thm:gen_form_bd} in Section \ref{sec:form_bound}, is the relative form bound itself. 

\begin{thm}\label{thm:Step2}
	Suppose Assumptions~\ref{ass:initial_model}--\ref{ass:bulkgap} and~\ref{ass:LTQO}--\ref{ass:short-range} hold, and fix $\Lambda\in\cP_0(\Gamma)$. Then, for any $\gamma \in (0,\gamma_0)$ and $|s|\leq s_0^\Lambda(\gamma)$, 
	there is a family of self-adjoint observables
	$\Phi^{(2)}(x,m,s)\in\cA_{b_x(m)}$, for each $x\in\Gamma$ and $m\geq R$, with the following properties:
	\begin{enumerate}
		\item[(i)] $\Phi^{(2)}(x,m,s)P_{b_x(m)} = P_{b_x(m)}\Phi^{(2)}(x,m,s)=0$;
		\item[(ii)] $\|\Phi^{(2)}(x,m,s)\|\leq 2sG^{(2)}_\Lambda(x,m)$ with
	\end{enumerate}
	\be
	G^{(2)}_\Lambda(x,m) = G_\Lambda(x,m/2)+2G_\Lambda^{(1)}(x,\lceil m/2\rceil)+2G_\Lambda^{(1)}(x,R)\sqrt{\lceil m/2\rceil^{\nu}G_0(m/2)}
	\ee
	where $G_\Lambda(x,m)$ is as in Theorem~\ref{thm:Step1}, $G_\Lambda^{(1)}(x,m) = \sum_{n\geq m}G_\Lambda(x,n)$, and $G_0$ is from Assumption~\ref{ass:LTQO}.
	Furthermore, $W(\Lambda,s)$ is given by the absolutely 
	convergent sum
	\be\label{final-decomp_W}
	W(\Lambda,s) = \sum_{x\in\Gamma}\sum_{m\geq R}\pi_0(\Phi^{(2)}(x,m,s)).
	\ee
\end{thm}

Note that the operator $W(\Lambda,s)$ is {\em a priori} defined in the GNS representation.
{\em A posteriori}, however, \eq{final-decomp_W} implies that  $W(\Lambda,s)$ is the image of  a quasi-local  observable in $\cA$. 

The decomposition from Theorem~\ref{thm:Step2} is proved in two steps. The first uses quasi-locality and conditional expectations to prove that for all $|s|\leq s_0^\Lambda(\gamma)$, the action of the spectral flow on the GNS Hamiltonian $H(\Lambda,s)$ can be again realized as a perturbation of $H_0$. Namely, we show that for all $\psi\in \pi_0(\cA_{\rm loc})\Omega$
\be \label{ql_pert}
\tilde{U}(\Lambda,s)^*H(\Lambda, s) \tilde{U}(\Lambda,s) \psi = H_0 \psi + \sum_{x\in\Gamma}\sum_{m\geq R} \tilde{\Phi}^{(1)}(x,m,s)\psi
\ee 
where the perturbation terms $\tilde{\Phi}^{(1)}(x,m,s)\in \pi_0(\cA_{b_x(m)})$ are self-adjoint, satisfy a norm bound that is linear in $s$, and are absolutely summable over $x\in \Gamma$ and $m\geq R$. This is accomplished in Theorem~\ref{thm:Step1} of Section \ref{sec:Step1} below.

In the second step, carried out in Section \ref{sec:Step2},  the final form of \eqref{final-decomp_W} from Theorem~\ref{thm:Step2} is proved using the frustration-free and LTQO ground state properties to produce a refined decomposition of the perturbation terms from \eqref{ql_pert}.

\subsection{Quasilocal decomposition of the transformed perturbation}\label{sec:Step1}

We now turn to establishing the first decomposition \eqref{ql_pert}, which is the content of the following theorem.

\begin{thm}\label{thm:Step1} 
	Under the conditions of Theorem \ref{thm:Step2}, there exists a function $G_\Lambda:\Gamma\times[0,\infty) \to [0,\infty)$ for which
	\be \label{G_summable}
	\sum_{x\in\Gamma}\sum_{m\geq R}G_\Lambda\left(x,m\right) <\infty
	\ee
	and a self-adjoint operator $\tilde{\Phi}^{(1)}(x,m,s)^* =  \tilde{\Phi}^{(1)}(x,m,s)\in \pi_0(\cA_{b_x(m)})$ for each $x\in \Gamma$ and $m \geq R$, such that $\|\tilde{\Phi}^{(1)}(x,m,s)\| \leq s G_\Lambda(x,m)$ and
	\begin{equation}\label{eq:step1_decomp}
		W(\Lambda,s)+E(\Lambda,s)\idty = \sum_{x\in \Gamma}\sum_{m\geq R} \tilde{\Phi}^{(1)}(x,m,s). 
	\end{equation}
	Moreover,  for each $x\in\Gamma$, the operator $\tilde{\Phi}^{(1)}(x,s) := \sum_{m\geq R} \tilde{\Phi}^{(1)}(x,m,s) $ belongs to $\pi_0(\cA)$ and commutes with the ground state projection $\kettbra{\Omega}{\Omega}$.
\end{thm}

The global term $\tilde{\Phi}^{(1)}(x,s)$ above will result from applying quasi-local maps $\cK_s^{i,\Lambda}$, $i=1,2$, to the interaction and perturbation terms associated to the site $x$. These maps are defined in terms of the examples introduced in Section~\ref{sec:ql}, and emerge from 
fixing any IAS $(\Lambda_n)$ and then applying Lemmas \ref{lem:FH_convergence}-\ref{lem:FsHs_convergence} to rewrite
\begin{equation}\label{finite_limit}
	(W(\Lambda,s)+E(\Lambda,s)\idty ) \psi = \lim_{n \to \infty} \left( \tilde{U}(\Lambda,s)^* \tilde{\cF}_s^{\Lambda}( \pi_0(H_{\Lambda_n} + s V_{\Lambda})) \tilde{U}(\Lambda,s) \psi - \tilde{\cF}(\pi_0(H_{\Lambda_n})) \psi \right)
\end{equation}
where we choose $\tilde{\cF}=\tilde{\cF}_0^\Lambda$. As the argument in the above limit is a finite sum of bounded operators, the various relationships \eqref{reps_of_wios}-\eqref{rep_spec_flow} between the quasi-local maps in the GNS representation to those on the $C^*$-algebra implies that for each $n$:
\begin{eqnarray} \label{pre_limit_quant}
	\tilde{\alpha}_s^\Lambda\circ\tilde{\cF}_s^{\Lambda}( \tilde{H}_{\Lambda_n} + s \tilde{V}_{\Lambda}) - \tilde{\cF}(\tilde{H}_{\Lambda_n}) & = & 
	\sum_{x \in \Lambda_n} \pi_0\left(\alpha_s^\Lambda\circ\cF_s^{\Lambda}(h_x)-\cF(h_x)\right) \nonumber \\ 
	&& \quad + \sum_{x\in\Lambda}\sum_{\substack{k\geq R: \\ b_x(k)\subseteq\Lambda}} \pi_0\left(s\alpha_s^\Lambda \circ \cF_s^{\Lambda}(\Phi(x,k))\right).
\end{eqnarray}  

Given this, for $i=1,2$ the map $\cK^{i, \Lambda}_s : \cA \to \cA$ are defined by
\begin{equation} \label{def:K1_K2}
	\cK_s^{1, \Lambda}(A) = \alpha_s^{\Lambda}( \cF_s^{\Lambda}(A)) - \cF(A) \quad \mbox{and} \quad \cK^{2, \Lambda}_s(A) = s \alpha_s^{\Lambda} ( \cF_s^{\Lambda}(A)). 
\end{equation}
It was proved, e.g. in \cite[Lemma 4.4]{nachtergaele:2022}, that both of these maps satisfy a local bound and quasi-local estimate that is independent of the finite volume $\Lambda$.
Specifically, for each $i=1,2$ there are non-negative numbers $p_i$, $q_i$ and $C_i$, and a decay function
$G_{\cK^i}$ (all independent of $\Lambda$) such that
\bea
\|\cK_s^{i, \Lambda}(A)\| & \leq & s C_i|X|^{p_i}\|A\| \label{Ki_locally_bounded} \\
\|[\cK_s^{i, \Lambda}(A),B]\| & \leq & s |X|^{q_i}\|A\|\|B\|G_{\cK^i}(d(X,Y)) \label{Ki_quasi_local}
\eea
hold for any $X,Y \in \mathcal{P}_0( \Gamma)$, $A \in \cA_X$, $B \in \cA_Y$, and $s \in \mathbb{R}$. In fact, one can take $p_1=q_1=2$, $p_2=0$ and $q_2=1$ and make explicit estimates on the decay function, see e.g. \cite[Remark 4.7]{nachtergaele:2022}. However, it suffices to note that each $G_{\cK^i}$ have finite moments of all orders in the sense of (\ref{all_finite_moms}).

However, as $\Lambda_n\uparrow\Gamma$ when $n\to\infty$, to prove that the decomposition in \eqref{eq:step1_decomp} is absolutely summable, we will need refinements 
% and quasi-local estimatesas in 
of \eqref{Ki_locally_bounded}-\eqref{Ki_quasi_local}  for $\cK_s^{1,\Lambda}$ that also decay in the distance $d(X,\Lambda)$. Both of these bounds
will be a consequence of the perturbation $V_\Lambda$ being locally supported, which implies that the spectral flow $\alpha_s^{\Lambda}$ is 
approximately the identity far from $\Lambda$.   The necessary bounds are the content of Lemmas~\ref{lem:lb_K1} and \ref{lem:ql_K1} below. 

%Given that we are working in the infinite volume setting, in the case of $\cK_{s}^{1,\Lambda}$ one needs a more refined version of
%the estimates in (\ref{Ki_locally_bounded}) - (\ref{Ki_quasi_local}) to obtain the absolute summability of the terms in (\ref{eq:step1_decomp}).We now prove the necessary distance-dependent locality bound and quasi-local estimates (namely, Lemma~\ref{lem:lb_K1} and Lemma~\ref{lem:ql_K1} below) which will be used in the proof of Theorem~\ref{thm:Step1}. 
%To be clear, we
%used bounds which have explicit decay with increasing distance from $\Lambda$.
%Generally, such bounds can not be expected, however, 
%Since the perturbation $V_{\Lambda}$ local support and
%$\alpha_s^{\Lambda}$ is approximately the identity far from $\Lambda$, we do obtain the following. We now prove the distance-dependent local bounds and quasi-local estimates for the map $\cK_s^{1,\Lambda}$ (namely, Lemma~\ref{lem:lb_K1} and Lemma~\ref{lem:ql_K1} below) which were used in the proof of Theorem~\ref{thm:Step1}. 

\begin{lemma}[Distance Locality Bound for $\cK_s^1$]\label{lem:lb_K1} 
	There exists a decay function $F_{\cK^1}$, with finite moments of all orders for which,
	given any $X,\Lambda \in \mathcal{P}_0( \Gamma)$ with $d(X, \Lambda) >0$,
	$A\in\cA_{X}$, and any $s\in\bR$,  the following local bound holds:
	\begin{equation}\label{close_to_id}
		\|\cK_s^{1, \Lambda}(A) \| \leq s|X|^2\|A\| F_{\cK^1}\left(d(X,\Lambda)\right).
	\end{equation}
\end{lemma}
It is easy to check that for fixed $\epsilon \in (0,1)$ and any decay function $F$ with finite $\nu$-moment, the function
$M_F^\epsilon:[0,\infty)\to[0,\infty)$ defined by
\be\label{moment_fn}
M_F^{(\epsilon)}(r) = \sum_{n\geq r} (n+1)^\nu F(\epsilon n)
\ee
is also a decay function. The proof of Lemma~\ref{lem:lb_K1} shows that one may take
\begin{equation} \label{our_bd_G1}
	F_{\cK^1}(r) = 2\kappa\|\Psi\|_{G_\Psi}\left(M_{G_\Psi}^{(1-\epsilon)}(r)+2G_\Psi(R)M_{G_\caF}^{(\epsilon)}(r)\right)
\end{equation}
where $G_{\cF}$ and $G_{\Psi}$ are the decay functions previously discussed in (\ref{wio_gen_ql_est}) and (\ref{est_spec_flow_int}).
Since $G_{\cF}$ and $G_{\Psi}$ both have finite moments of all orders, the same is true for $F_{\cK^1}$. 

The proof of Lemma~\ref{lem:lb_K1} will also make use of the following bound, which holds for any $F$ and $\epsilon$ as in \eqref{moment_fn}, and $\Lambda, X \in \mathcal{P}_0( \Gamma)$ such that $d(X, \Lambda)>0$:
\begin{equation}\label{moment_bound}
	\sum_{z\in\Lambda} F(\epsilon d(z,X)) \leq \kappa |X| M_F^{(\epsilon)}(d(X,\Lambda)).
\end{equation}
This follows from the following simple calculation
\begin{eqnarray}
	\sum_{z\in\Lambda} F(\epsilon d(z,X))  \leq \sum_{\substack{z\in\Gamma : \\ d(z,X)\geq d(X,\Lambda)}} F(\epsilon d(z,X)) 
	& \leq & \sum_{n\geq d(X,\Lambda)}\sum_{\substack{z\in\Gamma : \\ n\leq d(z,X)<n+1}} F(\epsilon d(z,X))  \nonumber \\
	&\leq & \kappa|X|\sum_{n\geq d(X,\Lambda)} (n+1)^\nu F(\epsilon n) 
\end{eqnarray}
where the last inequality uses that $|X(n)|\leq \kappa n^\nu|X|$ for any $n\geq1$ by $\nu$-regularity, see \eqref{Lambda_ball}.

\begin{proof}[Proof of Lemma~\ref{lem:lb_K1}:] Fix $X,\Lambda \in \cP_0(\Gamma)$ such that  $X \cap \Lambda = \emptyset$, and let $A \in \cA_X$ be arbitrary.
	Recall that $\cK_s^{1, \Lambda}$ is as defined in (\ref{def:K1_K2}), and that $D(\Lambda, s)$ from (\ref{SF_gen_alg}) is the generator of the spectral flow. Then, since $\alpha_0 = \rm{id}$ and $\cF =\cF_0^{\Lambda}$, it follows that
	\be\label{eq:FTOC}
	\cK_s^{1, \Lambda}(A)  = \int_0^s \frac{d}{dr} \alpha_r^{\Lambda}(\caF^{\Lambda}_r(A))dr = i\int_0^s \alpha^{\Lambda}_r([D(\Lambda,r),\caF^{\Lambda}_r(A)])dr
	\ee
	where one uses (\ref{para_der}) and \cite[Equation~(6.37)]{nachtergaele:2019} to obtain
	\begin{eqnarray}
		\frac{d}{ds} \cF_s^{\Lambda}(A) & = & i \int_{- \infty}^{\infty}  \int_0^t \tau_r^{(\Lambda, s)}([V_{\Lambda}, \tau_{t-r}^{(\Lambda, s)}(A)]) \, dr \, w_{\gamma}(t) \, dt \nonumber \\ & = & i \int_{- \infty}^{\infty}  \int_0^t \tau_r^{(\Lambda, s)}([V_{\Lambda}, \tau_{t-r}^{(\Lambda, s)}(A)]) \, dr \, \left( - \frac{d}{dt} W_{\gamma}(t) + \delta_0(t) \right) \, dt \nonumber \\
		& = & i  \int_{- \infty}^{\infty} \tau_t^{(\Lambda, s)}([V_{\Lambda}, A]) W_{\gamma}(t) \, dt  = 0 \, .
	\end{eqnarray}
	Here, the final two equalities follow from integrated by parts, and the fact that the supports of $V_\Lambda$ and $A$ are disjoint. 
	
	Returning to (\ref{eq:FTOC}), we expand the generator as in (\ref{def_spec_flow_int}) to write
	\begin{equation} \label{gen_expanded}
		[D(\Lambda, r), \cF_r^{\Lambda}(A) ] = \sum_{z \in \Lambda} \sum_{n \geq R} [ \Psi_{\Lambda}(z,n,r), \cF_r^{\Lambda}(A)] 
	\end{equation}
	Fix $\epsilon \in (0,1)$, and for each $z \in \Lambda$, set $k_z(\epsilon) = \epsilon d(z, X)$.  
	For each term in (\ref{gen_expanded}), we approximate  $ \cF_r^{\Lambda}(A)$ with a strictly local approximation:
	\begin{equation} \label{local_estimate}
		[ \Psi_{\Lambda}(z,n,r), \cF_r^{\Lambda}(A)] =  [ \Psi_{\Lambda}(z,n,r), \Pi_{X(k_z(\epsilon))}(\cF_r^{\Lambda}(A)) + \left( \cF_r^{\Lambda}(A) - \Pi_{X(k_z(\epsilon))}(\cF_r^{\Lambda}(A))\right)]
	\end{equation}
	where one uses conditional expectation associated with the inflated set $X(k_z(\epsilon))$, see \eqref{Lambda_ball}-\eqref{cond_exp_a}. 
	For the second term, one can apply the quasi-local bound for $\cF_r^{\Lambda}$ from (\ref{wio_gen_ql_est}) coupled with (\ref{cond_exp_est}) to produce
	\be
	\|[\Psi_\Lambda(z,n,r), \caF_r^{\Lambda}(A)-\Pi_{X(k_z(\epsilon))}(\caF_r^{\Lambda}(A))]\| \leq 4 \|A\||X|\|\Psi_\Lambda(z,n,r)\|G_{\caF}(k_z(\epsilon)) \, .
	\ee
	Then, summing over $z\in\Lambda$ and $n\geq R$, and  applying (\ref{est_spec_flow_int}) and (\ref{moment_bound}) gives the final estimate
	\be \label{lem_est_1}
	\sum_{z \in \Lambda} \sum_{n \geq R} \|[\Psi_\Lambda(z,n,r), \caF_r^{\Lambda}(A)-\Pi_{X(k_z(\epsilon))}(\caF_r^{\Lambda}(A))]\| \leq 4 \kappa \|A\||X|^2
	\| \Psi \|_{G_{\Psi}} G_{\Psi}(R) M_{G_{\caF}}^{(\epsilon)}(d(X,\Lambda)). 
	\ee
	
	To estimate the remaining terms in \eqref{local_estimate}, note that for each $z \in \Lambda$, $b_z(n) \cap X(k_z(\epsilon)) \neq \emptyset$ only when 
	$n \geq k_z(1- \epsilon)$. As a result, arguments similar to the prior estimate produce the bound
	\begin{eqnarray} \label{lem_est_2}
		\sum_{z\in\Lambda} \sum_{n\geq R} \| [\Psi_\Lambda(z,n,r), \Pi_{X(k_z(\epsilon))}(\caF^{\Lambda}_r(A))] \| & = & 
		\sum_{z \in\Lambda}\sum_{n\geq k_z(1-\epsilon)} \| [\Psi_\Lambda(z,n,r), \Pi_{X(k_z(\epsilon))}(\caF^{\Lambda}_r(A))] \| \nonumber \\
		& \leq & 2 \| A \| \| \Psi \|_{G_{\Psi}} \sum_{z \in \Lambda} G_{\Psi}(k_z(1- \epsilon)) \nonumber \\
		& \leq & 2 \kappa \| A \| |X| \| \Psi \|_{G_{\Psi}} M_{G_{\Psi}}^{(1- \epsilon)}(d(X, \Lambda)).
	\end{eqnarray}
	
	Recalling the specific decay function from (\ref{our_bd_G1}), the bound claimed in (\ref{close_to_id}) now follows
	by inserting (\ref{gen_expanded}) into (\ref{eq:FTOC}) and using the estimates found in (\ref{lem_est_1}) and (\ref{lem_est_2}) above. 
\end{proof}

By combining the estimate in Lemma~\ref{lem:lb_K1} and the original quasi-locality bound from \eqref{Ki_quasi_local}, one arrives that the following quasi-locality bound for $\| [\cK_s^{1, \Lambda}(A), B] \|$, which decays in both the distance between $X=\supp(A)$ and $Y=\supp(B)$ as well as the distance between $\Lambda$ and $X$. This is the content of the next lemma.

\begin{lemma}[Distance Quasi-Locality for $\cK_1$]\label{lem:ql_K1} 
	There exists a function $G:[0,\infty)\times [0,\infty) \to [0,\infty)$, non-increasing in both variables, such that
	given any $\Lambda, X,Y \in \mathcal{P}_0(\Gamma)$ with $d(X, \Lambda)>0$,  the bound 
	\be\label{distance_ql}
	\|[\cK_s^{1, \Lambda}(A),B]\| \leq s\|A\|\|B\||X|^2G\left(d(X,\Lambda),d(X,Y)\right)
	\ee
	holds for all $A \in \cA_X$, $B \in \cA_Y$, and $s \in \mathbb{R}$. More precisely, for any
	$\delta \in (0,1)$, one may choose
	\be \label{G2_decay}
	G(m,n) = \max\left\{ 2F_{\cK^1}^{\delta}(m)F_{\cK^1}^{1-\delta}(n),\, G_{\cK^1}^{\delta}(m)G_{\cK^1}^{1-\delta}(n)\right\}
	\ee
	where $F_{\cK^1}$  and $G_{\cK^1}$ are the decay functions from Lemma~\ref{lem:lb_K1} and \eqref{Ki_quasi_local}, and  $G^\delta(m):=(G(m))^\delta.$
\end{lemma}

In applications, it can be convenient to bound $G(n,m)$ by a function that separates over the two arguments. In this case, taking $\delta$ as in (\ref{G2_decay}),
\be \label{factorized_est}
G(m,n) \leq 2 F_{\delta}(m) \cdot F_{1- \delta}(n) \quad \mbox{with} \quad  F_{\delta}(m) = \max\{F_{\cK^1}^{\delta}(m), G_{\cK^1}^{\delta}(m)\}.\, 
\ee

\begin{proof}
	Fix $0<\delta<1$. In the case that $d(X,\Lambda) \leq d(X,Y)$, the quasi-locality estimate (\ref{Ki_quasi_local}) shows that
	\be\label{qlb_1}
	\|[\cK_s^{1, \Lambda}(A),B]\| \leq s|X|^2\|A\|\|B\|G_{\cK^1}^{\delta}(d(X,\Lambda))G_{\cK^1}^{1-\delta}(d(X,Y))  
	\ee
	where we have used that $G_{\cK^1}$ is non-increasing.
	
	Alternatively, if $d(X,\Lambda)>d(X,Y)$, the local bound from Lemma~\ref{lem:lb_K1} implies
	\be \label{qlb_2}
	\| [\cK_s^{1, \Lambda}(A),B]  \| \leq 2  \| \cK_s^{1, \Lambda}(A) \| \| B\|  \leq  2s|X|^2\|A\|\|B\| F_{\cK^1}\left(d(X,\Lambda)\right).
	\ee
	Since $F_{\cK^1}$ is also non-increasing, the bound $F_{\cK^1}(d(X, \Lambda)) \leq F_{\cK^1}^{\delta}\left(d(X,\Lambda)\right) F_{\cK^1}^{1-\delta}\left(d(X,Y)\right)$ follows. The bound (\ref{distance_ql}) is then a consequence of \eqref{qlb_1} and \eqref{qlb_2}.
\end{proof}

%We proceed with the proof of Theorem~\ref{thm:Step1} with the knowledge that such bounds do hold.
%, and recommend that the reader familiarize themselves with the statements of Lemma~\ref{lem:lb_K1}-\ref{lem:ql_K1} before proceeding.

We conclude this section by using Lemmas~\ref{lem:lb_K1}-\ref{lem:ql_K1} to prove Theorem~\ref{thm:Step1}.

\begin{proof}[Proof of Theorem~\ref{thm:Step1}] 

	Fix $\gamma \in (0, \gamma_0)$, $\Lambda \in \mathcal{P}_0( \Gamma)$,
	and take any IAS $(\Lambda_n)$ such that $\Lambda \subseteq \Lambda_n$ for all $n$. Define the spectral flow $\alpha_s^{\Lambda}$ and the weighted integral operators $\cF_s^{\Lambda}$, $\cF=\cF_0^\Lambda$ with respect to the choices of $\gamma$ and $\Lambda$ as in (\ref{spec_flow_on_A}) and (\ref{def:F+G}), and then take $\cK_s^{i,\Lambda}$, $i=1,2$, as defined \eqref{def:K1_K2}. 
	
	Let $\chi_{\Lambda}$ be the characteristic function of $\Lambda\subset\Gamma$. Then, for each $x \in \Gamma$ and $s\in \bR$ such that $|s| \leq s_0^{\Lambda}(\gamma)$, the self-adjoint operator $\tilde{\Phi}^{(1)}(x,s) = \pi_0(\Phi^{(1)}(x,s))\in\cB(\cH)$ is defined by
	\begin{equation} \label{phi_x_1}
		\Phi^{(1)}(x,s) = \cK_s^{1, \Lambda}(h_x) + \chi_{\Lambda}(x) \sum_{\stackrel{k \geq R:}{b_x(k) \subset \Lambda}} \cK_s^{2, \Lambda}( \Phi(x,k)) \in \cA.
	\end{equation}
	
	To show that each $\tilde{\Phi}^{(1)}(x,s)$ commutes with the ground state projection $\ketbra{\Omega}$, recall that the ground state of the perturbed system is
	$\Omega(\Lambda, s) = \tilde{U}(\Lambda,s) \Omega$ if $|s|\leq s_\gamma^\Lambda$. Then, recalling the relations \eqref{reps_of_wios}-\eqref{rep_spec_flow}, a simple calculation shows that for all $A \in \cA$
	\begin{eqnarray}
		\left[ \pi_0( \alpha_s^{\Lambda}( \cF_s^{\Lambda}(A))), \ketbra{\Omega}\right] & = & 
		\left[ \tilde{U}(\Lambda,s)^* \tilde{ \cF}_s^{\Lambda}(\pi_0(A)) \tilde{U}(\Lambda,s), \ketbra{\Omega}\right] \nonumber \\
		& = &  \tilde{U}(\Lambda,s)^* \left[ \tilde{ \cF}_s^{\Lambda}(\pi_0(A)), \ketbra{\Omega(\Lambda, s)}\right] \tilde{U}(\Lambda,s) = 0, \label{eq:partial_commutation}
	\end{eqnarray}
	where the final equality uses that (\ref{wi_commute_property}) holds since $|s| \leq s_0^{\Lambda}(\gamma)$.  Since \eqref{eq:partial_commutation} trivially holds for $s=0$, considering \eqref{def:K1_K2}, the above implies that
	\begin{equation}
		\left[ \pi_0( \cK_s^{1, \Lambda}(A)), \ketbra{\Omega} \right] = \left[ \pi_0( \cK_s^{2, \Lambda}(A)), \ketbra{\Omega} \right] = 0 \,.%\quad \mbox{for all } A \in \cA \mbox{ and } |s| \leq s_0^{\Lambda}(\gamma) \, .
	\end{equation}
	Hence, $[ \tilde{\Phi}^{(1)}(x,s), \ketbra{\Omega}]=0$ for all $x \in \Gamma$ and $|s| \leq s_0^{\Lambda}(\gamma)$ as claimed.
	
	To establish \eqref{eq:step1_decomp}, use the condition expectations from \eqref{delta_ops} to decompose each $\tilde{\Phi}^{(1)}(x,s)$ as
	\begin{equation} \label{phi_x_1_sum}
		\tilde{\Phi}^{(1)}(x,s) = \sum_{m \geq R} \tilde{\Phi}^{(1)}(x,m,s)
	\end{equation}
	where $\tilde{\Phi}^{(1)}(x,m,s) = \pi_0(\Phi^{(1)}(x,m,s) )\in  \pi_0(\cA_{b_x(m)})$ is defined for each $m \geq R$ by
	\begin{equation} \label{phi_x_1_bits}
		\Phi^{(1)}(x,m,s) := \Delta_{b_x(R)}^{m}(\cK_s^{1, \Lambda}(h_x)) + \chi_\Lambda(x)\sum_{\substack{R\leq k \leq m :\\ b_x(k)\subseteq \Lambda}} \Delta_{b_x(k)}^{m}(\cK_s^{2, \Lambda}(\Phi(x,k))) \,.
	\end{equation}
	With respect to this notation, \eqref{finite_limit}, \eqref{pre_limit_quant} and \eqref{phi_x_1} show that for all $\psi\in \pi_0(\cA^{\rm loc})\Omega$,
	\begin{eqnarray}
		\left( W(\Lambda, s) + E(\Lambda, s) \idty \right) \psi & = & \lim_{n \to \infty} \sum_{x \in \Lambda_n} \tilde{\Phi}^{(1)}(x,s) \psi \nonumber \\
		& = & \lim_{n \to \infty} \sum_{x \in \Lambda_n} \sum_{m \geq R} \tilde{\Phi}^{(1)}(x,m,s)  \psi.
	\end{eqnarray}
	Since this $\pi_0(\cA^{\rm loc})\Omega\subseteq \cH$ is dense, the equality in (\ref{eq:step1_decomp}) follows from establishing
	absolute summability of the terms $\tilde{\Phi}^{(1)}(x,m,s)$. This is achieved by defining a function function
	\be\label{G_Lambda}
	G_\Lambda(x,m) = \chi_{\Lambda(R)}(x) G_1(m) + \chi_{\Gamma\setminus\Lambda(R)}(x) G_2\left(d(b_x(R),\Lambda),m\right)
	\ee
	which bounds the norms of these terms and satisfies (\ref{G_summable}). Here, we note that $\Lambda(R)$ is as in \eqref{Lambda_ball}, and the functions $G_1:[0,\infty)\to[0,\infty)$ and $ G_2:[0,\infty)\times[0,\infty)\to[0,\infty)$ will be independent of $\Lambda$. 

	First, suppose $x \in \Gamma \setminus \Lambda(R)$. As $R \geq 0$ is the finite range of the unperturbed interaction, (\ref{phi_x_1_bits}) simplifies to
	\begin{equation}
		\Phi^{(1)}(x,m,s) = \Delta_{b_x(R)}^{m}(\cK_s^{1, \Lambda}(h_x)) \, .
	\end{equation}
	Then, applying Lemmas~\ref{lem:lb_K1} and \ref{lem:ql_K1} with the local approximation bound \eqref{ce_diff_bd} one finds 
	\begin{equation} \label{basic_bit_est}
		\| \tilde{\Phi}^{(1)}(x,m,s) \| = \| \Phi^{(1)}(x,m,s) \| \leq s G_2\left(d(b_x(R),\Lambda),m\right)
	\end{equation}
	where for any fixed $\delta \in (0, 1)$, the function $G_2$ can be taken to be
	\begin{equation} \label{G_far_est_1}
		G_2(l,m)  =   C \cdot \left\{ \begin{array}{cl}  F_{\cK^1}(l), & \mbox{if } m=R, \\ 
			4 F_{\delta}(l) \cdot F_{1- \delta}(m-R-1), & \mbox{if } m \geq R+1 .\end{array} \right.
	\end{equation}
	Here, $C = \kappa^2 R^{2 \nu} \| h \|_{\infty}$, $F_{\cK^1}$ is the function from Lemma~\ref{lem:lb_K1}, and
	\be
	F_\delta(l) = \max\left\{(F_{\cK^1}(l))^\delta,(G_{\cK^1}(l))^\delta\right\}.
	\ee
	More specifically, the bound in \eqref{basic_bit_est} for $m=R$ is a direct application of Lemma~\ref{lem:lb_K1} while the bound for $m \geq R+1$ follows from the quasi-local estimate in Lemma~\ref{lem:ql_K1} and the subsequent bound (\ref{factorized_est}) coupled with \eqref{def:qlm}-\eqref{ce_diff_bd}.
	
	Given \eqref{G_Lambda}, the summability of $G_\Lambda$ over the sites $x\in\Gamma\setminus\Lambda(R)$ follows from observing that
	\begin{align}
		\sum_{x\in\Lambda(R)^c} \sum_{m\geq R}G_2\left(d(b_x(R),\Lambda),m\right) & =  C\sum_{x\in\Lambda(R)^c}F_{\cK^1}(d(b_x(R),\Lambda)) \nonumber\\ & + 4C\sum_{x\in\Lambda(R)^c}F_{\delta}(d(b_x(R),\Lambda)) \sum_{m\geq R+1}F_{1-\delta}(m-R-1)<\infty
	\end{align}
	as both $F_{\cK^1}$ and $G_{\cK^1}$ (and, thus, $F_\delta$) have finite moments of all orders. In particular, for any decay function $F:[0,\infty)\to [0,\infty)$ with a finite $\nu$-moment,
	\[
	\sum_{x\in\Lambda(R)^c}F(d(b_x(R),\Lambda) )\leq \sum_{n\geq R}\sum_{x\in\Lambda(n+1)\setminus\Lambda(n)} F(n-R) \leq \kappa|\Lambda|\sum_{n\geq R}(n+1)^\nu F(n-R) < \infty.
	\]
	
	We now turn to the sites $x \in \Lambda(R)$, for which we demonstrate that 
	\be \label{close_goal}
	\|\Phi^{(1)}(x,m,s)\| \leq sG_1(m)
	\ee
	where $G_1$ is a summable function. First consider (\ref{phi_x_1_bits}) when $m=R$. Combining the local bounds (\ref{Ki_locally_bounded}), the uniform bound \eqref{uni_bd}, and the interaction bound in Assumption~\ref{ass:short-range}, one produces the $x$-independent bound
	\begin{eqnarray} \label{m=R_est}
		\| \Phi^{(1)}(x,R,s) \| & \leq & \| \cK^{1, \Lambda}(h_x) \| + \| \cK^{2, \Lambda}(\Phi(x,R)) \| \nonumber \\
		& \leq & s C_1 \kappa^2 R^{2 \nu} \| h \|_{\infty} + s C_2 \| \Phi \| e^{- a R^{\theta}}\,.
	\end{eqnarray}
	
	Alternatively, for $m\geq R+1$,  (\ref{phi_x_1_bits}) can be estimated as 
	\begin{eqnarray} \label{trivial_bound}
		\|\Phi^{(1)}(x,m,s)\| & \leq & \|\Delta_{b_x(R)}^{m}(\cK_s^{1, \Lambda}(h_x))\| + \chi_\Lambda(x)\sum_{\substack{R\leq k \leq m :\\ b_x(k)\subseteq \Lambda}} \|\Delta_{b_x(k)}^{m}(\cK_s^{2, \Lambda}(\Phi(x,k)))\| \nonumber \\
		& \leq & 2 s \kappa^2 R^{2 \nu} \| h \|_{\infty} G_{\cK^1}(m-R-1) + 2 s \kappa \sum_{\substack{R\leq k \leq m :\\ b_x(k)\subseteq \Lambda}} k^{\nu} \| \Phi(x, k) \| 
		G_{\cK^2}(m-k-1)
	\end{eqnarray}
	where one uses the quasi-local estimates from (\ref{Ki_quasi_local}) and the local approximation
	bound in (\ref{ce_diff_bd}). Given Assumption~\ref{ass:short-range}, the final sum above can be further estimated as
	\begin{eqnarray}
		\sum_{\substack{R\leq k \leq m :\\ b_x(k)\subseteq \Lambda}} k^{\nu} \| \Phi(x, k) \| 
		G_{\cK^2}(m-k-1) & \leq & \| \Phi \| \sum_{k=R}^m k^{\nu} e^{-a k^{\theta}} G_{\cK^2}(m-k-1) \nonumber \\
		& \leq & \| \Phi \| \left( G_{\cK^2}(m/2) \sum_{k=R}^{m/2 -1} k^{\nu} e^{-a k^{\theta}}  + G_{\cK^2}(0) \sum_{k \geq m/2} k^{\nu} e^{-a k^{\theta}} \right).
	\end{eqnarray}
	
	To simplify notation, let $	M_{\Phi}(r) := \sum_{k \geq r} k^{\nu} e^{- a k^{\theta}} $ denote the $\nu$-th moment of the decay function associated with the perturbation $\Phi$ from Assumption~\ref{ass:short-range}. Then, in summary, one has that for $x\in\Lambda(R)$, \eqref{close_goal} holds for the decay function $G_1$ defined by
	\begin{equation}
		G_1(R) =   C_1 \kappa^2 R^{2 \nu} \| h \|_{\infty} +  C_2 \| \Phi \| e^{- a R^{\theta}}
	\end{equation}
	and for $m \geq R+1$,
	\begin{equation}\label{long_decay}
		G_1(m) =  2 C G_{\cK^1}(m-R-1) + 2 \kappa \|\Phi\| \left(M_\Phi(R)G_{\cK^2}(m/2)+G_{\cK^2}(0)M_\Phi(m/2)\right).
	\end{equation}

	Since each of the decay functions in \eqref{long_decay}
	has finite moments of all orders, it is clear that $\sum_{m\geq R}G_1(m) <\infty$. As a consequence, $G_\Lambda$ as in \eqref{G_Lambda} satisfies
	\[
	\sum_{x\in\Lambda(R)}\sum_{m\geq R} G_\Lambda(x,m) \leq \kappa R^\nu|\Lambda| \sum_{m\geq R} G_1(m) <\infty.
	\]
	This demonstrates absolute summability of the terms in (\ref{eq:step1_decomp}), and hence, completes the proof of
	Theorem~\ref{thm:Step1}. 
\end{proof}

%Let us first recall some local bounds and quasi-locality estimates which we proved in \cite[Lemma 4.4]{}.
%Corresponding to the maps $\cK^{i, \Lambda}$, as defined in (\ref{def:K1_K2}) for $i=1,2$, 
%there are numbers $p_i$ and $C_i$ (both independent of $\Lambda$) for which the local bound:
%\be\label{Ki_locally_bounded}
% \|\cK_s^{i, \Lambda}(A)\| \leq s C_i|X|^{p_i}\|A\| 
% \ee
%holds for any $X \in \mathcal{P}_0( \Gamma)$, $A \in \cA_X$, and $s \in \mathbb{R}$. In fact, $p_1=2$ and $p_2=0$.
%In addition, the following quasi-locality bound was also established. There are numbers $q_i$ and decay functions
%$G_{\cK^i}$ (both independent of $\Lambda$) for which given $X, Y \in \mathcal{P}_0(\Gamma)$, the bound
% \be\label{Ki_quasi_local}
% \|[\cK_s^{i, \Lambda}(A),B]\| \leq s |X|^{q_i}\|A\|\|B\|G_{\cK^i}(d(X,Y)) 
%\ee
%holds for all $A \in \cA_X$, $B \in \cA_Y$, and $s \in \mathbb{R}$. In fact, $q_1=2$ and $q_2=1$. One can make explicit estimates
%on the decay rate for these functions $G_{\cK^i}$, see e.g. \cite[Remark 4.7]{}, however for our purposes
%it suffices to note that they each have finite moments of all orders, in the sense of (\ref{all_finite_moms}).

\subsection{The final decomposition of the transformed perturbation via LTQO}\label{sec:Step2}
We now turn our attention to proving Theorem~\ref{thm:Step2}, which is a consequence of one last decomposition of the transformed perturbation from Theorem~\ref{thm:Step1}, i.e.
\[\left\{\tilde{\Phi}^{(1)}(x,m,s)\in\pi_0(\caA_{b_x(m)}): x\in \Gamma, \, m\geq R\right\}.\] 
The key component for proving the desired norm bounds for this final decomposition is Lemma~\ref{lem:LTQO_Est} below, and it is in the proof of this result where one needs the LTQO property from Assumption~\ref{ass:LTQO}. To this end, we first shift the transformed perturbation terms by their expectation in the ground state $\Omega$, as this will put us in the appropriate setting to apply LTQO.

Throughout this section, we assume $\gamma\in(0,\gamma_0)$ is fixed and that $s\in\bR$ is such that $|s|\leq s_0^\Lambda(\gamma)$. As such, $\Omega(\Lambda,s) = \tilde{U}(\Lambda,s)\Omega$ is the ground state of $H(\Lambda,s)$, and one finds that $ W(\Lambda, s) \Omega =0$ from considering (\ref{main_decomp}) in the case $\psi=\Omega$. Thus, Theorem~\ref{thm:Step1} implies that for any $\psi \in \pi_0( \cA^{\rm loc}) \Omega$
\be
W(\Lambda,s)\psi=\tilde{U}(\Lambda,s)^* H(\Lambda,s) \tilde{U}(\Lambda,s) \psi - H_0\psi - E(\Lambda,s) \psi  =
\sum_{x\in\Gamma}\sum_{m\geq R} \tilde{\Phi}^{(1)}_\omega(x,m,s)\psi,
\ee
where the (self-adjoint) observables $\tilde{\Phi}^{(1)}_\omega(x,m,s) \in \pi_0(\caA_{b_{x}(m)})$ are defined by
\be\label{gns_local_at_x_c}
\tilde{\Phi}^{(1)}_\omega(x,m,s)=\tilde{\Phi}^{(1)}(x,m,s)-\langle\Omega,\tilde{\Phi}^{(1)}(x,m,s)\Omega\rangle\idty \,.
\ee
%The analogous relation in $\cA$ holds, namely
%\be
%\Phi^{(1)}_\omega(x,m,s)=\Phi^{(1)}(x,m,s)-\omega_0(\Phi^{(1)}(x,m,s))\idty.
%\ee
For the proofs of Lemma~\ref{lem:LTQO_Est} and Theorem~\ref{thm:Step2}, it is also convenient to set
\be \label{gns_global_at_x_c}
\tilde{\Phi}^{(1)}_{\omega}(x,s) := \sum_{m \geq R} \tilde{\Phi}_{\omega}^{(1)}(x,m,s), \qquad x\in \Gamma
\ee
which belongs to $\pi_0(\caA)$ by Theorem~\ref{thm:Step1}.

\begin{lemma}\label{lem:LTQO_Est}
	Let $\tilde{P}_{b_x(n)}= \pi_0(P_{b_x(n)})\in\cB(\cH)$ denote the representation of the ground state projection $P_{b_x(n)}$ in the GNS space. Then, under the assumptions of Theorem~\ref{thm:Step2}, the bound
	\be \label{LTQO_est_b}
	\left\| \sum_{k=R}^m \tilde{\Phi}_\omega(x,k,s)\tilde{P}_{b_x(n)}\right\| \leq 2s\left(G_\Lambda^{(1)}(x,m+1)+G_\Lambda^{(1)}(x,R)\sqrt{(1+m)^{\nu}G_0(n-m)}\right)
	\ee
	holds where $G_0$ is the decay function from Assumption~\ref{ass:LTQO}, and $G_\Lambda^{(1)}(x,m) = \sum_{k\geq m}G_\Lambda(x,k)$ with $G_\Lambda$ as in Theorem~\ref{thm:Step1}.
\end{lemma}

\begin{proof}
	To begin, one uses the LTQO property \eqref{LTQOforA} to show that
	\be\label{LTQO_cor}
	\left|\|\tilde{A}\tilde{P}_{b_x(n)}\|-\|\tilde{A}P_\Omega\|\right| \leq \|A\|\sqrt{(1+m)^{\nu}G_0(n-m)} \; \mbox{  for all  } \; A\in \cA_{b_x(m)}
	\ee
	where $P_\Omega = \ketbra{\Omega}$ and $\tilde{A} = \pi_0(A)$. To see this, first note that the inequality
	$|a-b|^2 \leq |a^2-b^2|$ for any $a,b\geq 0,$ implies that
	\[
	\left|\|\tilde{A}\tilde{P}_{b_x(n)}\|-\|\tilde{A}P_\Omega\|\right|^2 \leq \left|\|\tilde{A}\tilde{P}_{b_x(n)}\|^2-\|\tilde{A}P_\Omega\|^2\right| .
	\]
	Recalling that $(\cH,\pi_0, \Omega)$ is the GNS representation of the unperturbed ground state $\omega_0$, the second term on the right-hand-side above is simply
	\[
	\|\tilde{A}P_\Omega\|^2 = \braket{\Omega}{\pi_0(A^*A)\Omega}\|P_\Omega\| = \omega_0(A^*A)\|P_{b_x(n)}\|.
	\]
	%where we used $\omega_0(A^*A)\geq 0$ and that the norm of any nonzero orthogonal projection is one. 
	As $\pi_0$ is norm-preserving, it also follows that
	\[
	\|\tilde{A}\tilde{P}_{b_x(n)}\|^2 = \|\pi_0(AP_{b_x(n)})\|^2 = \|P_{b_x(n)}A^*AP_{b_x(n)}\|.
	\]
	Given these observations, Assumption~\ref{ass:LTQO} then implies that
	\begin{align}
		\left|\|\tilde{A}\tilde{P}_{b_x(n)}\|-\|\tilde{A}P_\Omega\|\right|^2 & \leq \left|\|P_{b_x(n)}A^*AP_{b_x(n)}\|-\omega_0(A^*A)\|P_{b_x(n)}\|\right| \nonumber\\
		& \leq \|P_{b_x(n)}A^*AP_{b_x(n)}-\omega_0(A^*A)P_{b_x(n)}\| \nonumber \\
		& \leq \|A\|^2 (1+m)^{\nu}G_0(n-m),
	\end{align}
	which establishes (\ref{LTQO_cor}).
	
	Now using (\ref{LTQO_cor}) with $\tilde{A} = \sum_{k=R}^m  \tilde{\Phi}_\omega^{(1)}(x,k,s) \in \pi_0(\cA_{b_x(m)})$, one finds that for any $n\geq m$, 
	\begin{eqnarray}\label{Bound1}
		\left\|\sum_{k=R}^m  \tilde{\Phi}_\omega^{(1)}(x,k,s) \tilde{P}_{b_x(n)}\right\| 
		& \leq & \left\|\sum_{k=R}^m  \tilde{\Phi}_\omega^{(1)}(x,k,s) P_{\Omega}\right\| + \left\|\sum_{k=R}^m  \tilde{\Phi}_\omega^{(1)}(x,k,s) \right\| \sqrt{(1+m)^{\nu}G_0(n-m)} \nonumber \\
		& \leq & \left\|\sum_{k=R}^m  \tilde{\Phi}_\omega^{(1)}(x,k,s)P_{\Omega}\right\| + 2sG_\Lambda^{(1)}(x,R)\sqrt{(1+m)^{\nu}G_0(n-m)} \label{Bound1-2}
	\end{eqnarray}
	where the last inequality follows from Theorem~\ref{thm:Step1} as $\|\tilde{\Phi}_\omega^{(1)}(x,k,s)\| \leq 2\|\tilde{\Phi}^{(1)}(x,k,s)\| \leq 2 s G_\Lambda(x,k)$.
	
	The remaining operator norm from \eqref{Bound1} can then be trivially bounded in terms of $\tilde{\Phi}_\omega^{(1)}(x,s)$ from (\ref{gns_global_at_x_c}) as follows:
	\be 
	\left\|\sum_{k=R}^m  \tilde{\Phi}_\omega^{(1)}(x,k,s) P_{\Omega}\right\| \leq \left\| \tilde{\Phi}_{\omega}^{(1)}(x,s) P_{\Omega}\right\|
	+ 
	\sum_{k= m+1}^\infty \left\| \tilde{\Phi}_\omega^{(1)}(x,k,s) P_{\Omega}\right\|.
	\ee
	Once again applying Theorem~\ref{thm:Step1} then shows that
	\be\label{Bound2}
	\sum_{k= m+1}^\infty \left\| \tilde{\Phi}_\omega^{(1)}(x,k,s) P_{\Omega}\right\| \leq 2\sum_{k= m+1}^\infty \left\| \tilde{\Phi}^{(1)}(x,k,s)\right\| \leq 2sG_\Lambda^{(1)}(x,m+1),
	\ee
	and, moreover, $[\tilde{\Phi}_{\omega}^{(1)}(x,s), P_{\Omega}]= [\tilde{\Phi}^{(1)}(x,s), P_{\Omega}]=0$. As a result, $ \left\| \tilde{\Phi}_{\omega}^{(1)}(x,s) P_{\Omega}\right\|=0$ since
	\begin{align}\label{Bound3}
		\tilde{\Phi}_{\omega}^{(1)}(x,s)P_\Omega  = P_\Omega\tilde{\Phi}_{\omega}^{(1)}(x,s)P_\Omega = \braket{\Omega}{\tilde{\Phi}_{\omega}^{(1)}(x,s)\Omega} P_\Omega=0,
	\end{align}
	where the last inequality follows from \eqref{gns_local_at_x_c}-\eqref{gns_global_at_x_c}. Thus, inserting \eqref{Bound2} into \eqref{Bound1-2} proves \eqref{LTQO_est_b}. 
\end{proof}

We now prove Theorem~\ref{thm:Step2}, which uses both Lemma~\ref{lem:LTQO_Est} and the frustration-free property.
\begin{proof}[Proof of Theorem \ref{thm:Step2}]
	Fix $x\in\Gamma$ and recall that $P_\Omega = \ketbra{\Omega}$. Since $[\tilde{\Phi}_{\omega}^{(1)}(x,s), P_\Omega]=0$, one can write
	\[
	\tilde{\Phi}_{\omega}^{(1)}(x,s) = P_\Omega\tilde{\Phi}_{\omega}^{(1)}(x,s)P_\Omega
	+(\idty-P_\Omega)\tilde{\Phi}_{\omega}^{(1)}(x,s)(\idty-P_\Omega) 
	=(\idty-P_\Omega)\tilde{\Phi}_{\omega}^{(1)}(x,s)(\idty-P_\Omega) ,
	\]
	where the last equality uses \eqref{Bound3}. The terms $\tilde{\Phi}^{(2)}(x,m,s)$ are defined by decomposing $\idty-P_\Omega$ in terms of the finite volume ground state projections $\tilde{P}_{n}:=\tilde{P}_{b_x(n)}\in\pi_0(\cA_{b_x(n)})$. 
	
	First, note that $\tilde{P}_{n}$ converges strongly to $P_\Omega$ for all $\psi\in\cH$ by the frustration-free and LTQO properties. 
	As a consequence, the collection of operators
	\begin{align}
		\tilde{E}_n = \begin{cases} \idty - \tilde{P}_{R}, & n= R\\
			\tilde{P}_{n-1}-\tilde{P}_{n}, & n\geq R+1
		\end{cases}
	\end{align}
	forms a family of orthogonal projections that are mutually orthogonal and sum to $\idty - P_\Omega$. That is,
	\be\label{En_props}
	\tilde{E}_n^* = \tilde{E}_n, \quad \tilde{E}_n\tilde{E}_m = \delta_{m,n}\tilde{E}_n, \quad\text{and}\quad (\idty-P_\Omega)\psi = \sum_{n\geq R} \tilde{E}_n\psi \quad \forall \; \psi\in\cH,
	\ee
	where the second equality holds since the frustration-free property implies $\tilde{P}_n\tilde{P}_m = \tilde{P}_m$ for $m\geq n$. Moreover, it trivially holds that
	\be\label{En_finite_sum}
	\idty - \tilde{P}_m = \sum_{n=R}^m \tilde{E}_n.
	\ee
	Using (\ref{gns_global_at_x_c}), the above properties imply that for all $\psi,\phi\in\cH$,
	\begin{align}
		\braket{\phi}{(\idty-P_\Omega)\tilde{\Phi}_{\omega}^{(1)}(x,s)(\idty-P_\Omega) \psi}
		= \sum_{k,m,n\geq R}\braket{\phi}{\tilde{E}_n\tilde{\Phi}_\omega^{(1)}(x,k,s)\tilde{E}_m\psi}
	\end{align}
	We note that the triple sum of operators actually converges absolutely in norm, and so the operator equality holds in the norm sense. 
	
	Each term $\Phi^{(2)}(x,m,s) \in \caA_{b_x(m)}$ will be defined as a sum of two self-adjoint terms
	\[
	\pi_0(\Phi^{(2)}(x,m,s)) = \Theta_1(x,m,s) + \Theta_2(x,m,s)
	\]
	each of which is annihilated by the ground state projection $\tilde{P}_m$. Fix $k\geq R$, and use the properties in \eqref{En_props}-\eqref{En_finite_sum} to write
	\be
	\sum_{n,m\geq R} \tilde{E}_n\tilde{\Phi}^{(1)}_\omega(x,k,s)\tilde{E}_m = (\idty - \tilde{P}_{2k})\tilde{\Phi}^{(1)}_\omega(x,k,s)(\idty - \tilde{P}_{2k}) + \sum_{m>2k} \tilde{\Phi}_{k,m},
	\ee\label{Phi2_initial_decomp}
	where $\tilde{\Phi}_{k,m}=\tilde{\Phi}_{k,m}^*$ is defined by
	\be\label{Phi_km}
	\tilde{\Phi}_{k,m} = \tilde{E}_m\tilde{\Phi}^{(1)}_\omega(x,k,s)(\idty-\tilde{P}_{m-1}) + (\idty-\tilde{P}_m)\tilde{\Phi}^{(1)}_\omega(x,k,s)\tilde{E}_m\in\pi_0(\cA_{b_x(m)}).
	\ee
	Self-adjointness follows from noting that $\idty-\tilde{P}_m = \idty-\tilde{P}_{m-1} + \tilde{E}_m$. %, and note that $\tilde{\Phi}_{k,m}\in\pi_0(\cA_{b_x(m)})$.
	%follows from our assumptions on $f$ since these guarantee that $g:\bZ_{\geq R} \to \bZ_{\geq R}$ is strictly increasing with $g(k)\geq k$ for all $k\in\bZ_{\geq R} := \bZ\cap[R,\infty)$.
	
	For each $m\geq R$, define $\Theta_1(x,m,s) \in\pi_0(\cA_{b_x(m)}),$ by
	\be\label{def:Theta_1}
	\Theta_1(x,m,s) = 
	\begin{cases}
		(\idty - \tilde{P}_{m})\tilde{\Phi}^{(1)}_\omega(x,m/2,s)(\idty - \tilde{P}_{m}), & m \mbox{  even } \\ 
		0, & m \mbox{  odd }
	\end{cases}
	\ee
	These operators are self-adjoint, satisfy $\Theta_1(x,m,s)\tilde{P}_{m}=0$, and Theorem~\ref{thm:Step1} implies that their norm is bounded from above by $2sG_\Lambda(x,m/2)$ as for $m$ even:
	\be\label{Theta_1_est}
	\|\Theta_1(x,m,s)\| \leq \|\Phi^{(1)}_\omega(x,m/2,s)\| \leq 2sG_\Lambda(x,m/2). %\leq 2sG(x,f(g(k))).
	\ee
	
	For the $\Theta_2$ terms, one sums the remaining terms $\sum_{m>2k}\tilde{\Phi}_{k,m}$ over $k$, and then uses the indicator function $\chi_{m>2k}$ to exchange the summations as follows:
	\begin{align*}\label{def:Theta_2}
		\sum_{k\geq R} \sum_{m>2k}\tilde{\Phi}_{k,m} & = \sum_{k\geq R}\sum_{m>2R}\tilde{\Phi}_{k,m}\chi_{\{m>2k\}} \\ 
		& = \sum_{m>2R}\sum_{R\leq k < m/2}\tilde{\Phi}_{k,m} \\
		& = \sum_{m>2R}\Theta_2(x,m,s) 
	\end{align*}
	where, for $m>2R$ one recalls \eqref{Phi_km} and defines
	\be
	\Theta_2(x,m,s) = \sum_{k=R}^{\lceil m/2 \rceil -1} \tilde{E}_m\tilde{\Phi}^{(1)}_\omega(x,k,s)(\idty-\tilde{P}_{m-1}) + (\idty-\tilde{P}_m)\tilde{\Phi}^{(1)}_\omega(x,k,s)\tilde{E}_m.
	\ee
	This definition is extended by setting $\Theta_2(x,m,s)=0$ for $R\leq m \leq 2R$. One sees that $\Theta_2(x,m,s)\tilde{P}_m=0$ by applying $\tilde{P}_m\tilde{P}_{m-1}=\tilde{P}_m$ to verify 
	\[\tilde{E}_m\tilde{P}_m = (\tilde{P}_{m}-\tilde{P}_{m-1})\tilde{P}_m=0.\]
	Finally, applying Lemma~\ref{lem:LTQO_Est},
	\begin{align}
		\|\Theta_2(x,m,s)\| & \leq 2\left\| \sum_{k=R}^{\lceil m/2 \rceil -1} \tilde{\Phi}^{(1)}_\omega(x,k,s)(\idty-\tilde{P}_{m-1}) \right\| \nonumber\\
		& \leq 4s\left(G_\Lambda^{(1)}(x,\lceil m/2 \rceil )+G_\Lambda^{(1)}(x,R)\sqrt{\lceil m/2 \rceil ^{\nu}G_0(m/2)}\right) \label{Theta_2_est}
	\end{align}
	where $G_\Lambda^{(1)}(x,m) = \sum_{n\geq m} G_\Lambda(x,n)$.
	
	Therefore, setting $\tilde{\Phi}^{(2)}(x,m,s) = \Theta_1(x,m,s)+\Theta_2(x,m,s)$ for all $m$, we have constructed self-adjoint operators such that
	\be\label{phi2_decomp}
	W(\Lambda,s) = \sum_{x\in\Gamma}\sum_{m\geq R} \tilde{\Phi}^{(2)}(x,m,s).
	\ee
	Moreover, these operators satisfy that for all $m\geq R$ and $x\in\Gamma$,
	\[
	\tilde{\Phi}^{(2)}(x,m,s)\tilde{P}_m = \tilde{P}_m\tilde{\Phi}^{(2)}(x,m,s) = 0
	\]
	and, combining \eqref{Theta_1_est} and \eqref{Theta_2_est}, their norms can be bounded above by
	\be \label{Phi2_term_bound}
	\|\tilde{\Phi}^{(2)}(x,m,s)\| \leq 2sG_\Lambda^{(2)}(x,m)
	\ee
	where
	\be
	G_\Lambda^{(2)}(x,m) = G_\Lambda(x,m/2)+2G_\Lambda^{(1)}(x,\lceil m/2 \rceil )+2G_\Lambda^{(1)}(x,R)\sqrt{\lceil m/2 \rceil^{\nu}G_0(m/2)}.
	\ee
	
	The absolute summability of the series in \eqref{phi2_decomp} is a direct consequence of $G_0$ being summable as well as that both $G_{\Lambda}$ and $G_{\Lambda}^{(1)}$ satisfy \eqref{G_summable}. For $G_\Lambda^{(1)}$ this can easily be seen from the fact that $G_\Lambda$ is a combination of functions with finite moments of all orders, see specifically \eqref{G_Lambda}, \eqref{G_far_est_1} and \eqref{long_decay}.
\end{proof}

\section{Proving Theorems  \ref{thm:finite_volume}  and \ref{thm:infinite_volume} via a form bound for the GNS Hamiltonian}\label{sec:form_bound}

In this section, we generalize \cite[ Theorem 3.8]{nachtergaele:2022}, which was itself based off \cite[Proposition 2]{michalakis:2013}, so that it is applicable to the setting of infinite systems in their GNS representation. Afterwards, we apply the form bound in conjunction with Theorem~\ref{thm:Step2} to prove Theorem~\ref{thm:finite_volume}. We then conclude with the proof of Theorem~\ref{thm:infinite_volume}, which follows as a consequence of Theorem~\ref{thm:finite_volume}.

\subsection{The Michalakis-Zwolak relative form bound}
The form bound result is stated in the setting described in Section \ref{sec:main} and, in particular, under Assumption~\ref{ass:localgaps}. Thus, there is a family of sets \[\mathcal{S} = \left\{ \Lambda(x,n) |  x \in \Gamma, n \geq 0 \mbox{ s.t. } b_x(n)\subseteq \Lambda(x,n) \right\},\] accompanied with
a family $\mathcal{T} = \{ \mathcal{T}_n : n \geq 0\}$ of separating partitions of $(c , \zeta)$-polynomial growth,
for which the associated finite volume Hamiltonians satisfy
\be \label{gap_gam_n}
H_{\Lambda(x,n)} \geq \gamma(n) P_{\Lambda(x,n)}, \mbox{ for all } n\geq R \, .
\ee
Moreover, the local gaps are further assumed to satisfy $\gamma(n) \geq \gamma_1/n^\alpha$
for some $\gamma_1>0$ and $\alpha\geq 0$.

\begin{thm}[Michalakis-Zwolak \cite{michalakis:2013}]\label{thm:gen_form_bd}
	%	Let $(\Gamma, d)$ be a $\nu$-regular metric space and $\mathcal{A}$ an associated
	%	quantum spin $C^*$-algebra of quasi-local observables. 
	%	Let $\{ h_x \}_{x \in \Gamma}$ denote a uniformly bounded, finite-range, frustration free 
	%	interaction on $(\Gamma, d)$, see Section~\ref{sec:main} for details.  Let $\omega$ be an 
	%	infinite-volume, zero-energy ground state for the dynamics associated to this interaction
	%	and denote by $(\mathcal{H}_{\omega}, \pi_{\omega}, \Omega_{\omega})$ the corresponding
	%	GNS-triple.  Suppose that there is a collection of subsets $\mathcal{S} \subset \mathcal{P}_0( \Gamma)$, labeled as in
	%	(\ref{S}), for which this interaction is locally gapped, in the sense of (\ref{local_gap}), and to which
	%	there is a family of separating partitions of polynomial growth. 
	Let $H_0$ be the GNS Hamiltonian associated with a zero-energy ground state of an initial system satisfying Assumptions~\ref{ass:initial_model} and \ref{ass:localgaps}, and let $V\in\cA$ be a perturbation associated with an absolutely-summable, anchored interaction on $(\Gamma, d)$. That is, there exist $\Phi(x,n)^* = \Phi(x,n) \in \mathcal{A}_{b_x(n)} $ for all $x\in \Gamma$ and $n\geq R$ such that
	\begin{equation}
		V = \sum_{x \in \Gamma} \sum_{n \geq R} \Phi(x,n), \qquad  \sum_{x \in \Gamma}\sum_{n \geq R} \| \Phi(x,n) \| < \infty. \label{def_V}	
	\end{equation}
	In addition, assume that terms of $V$ annihilate the finite-volume ground state projections of the initial system, i.e.
	\begin{equation} \label{terms_annih}
		\Phi(x,n)P_{b_x(n)} = P_{b_x(n)} \Phi(x,n) = 0 \quad \mbox{for all } x \in \Gamma \mbox{ and } n \geq R \, ,
	\end{equation}
	and that there is a decay function $G$ with finite $(\zeta+\alpha)$-moment such that $	\sup_{x \in \Gamma} \| \Phi(x,n) \| \leq G(n) \, .$
	Then, for all $\psi \in \dom{H_0}$,
	\begin{equation}
		\label{V_form_bd}
		\left| \braket{\psi}{\pi_{\omega}(V) \, \psi} \right|
		\; \leq \;
		\beta
		\braket{\psi}{H_0\, \psi} \quad \text{where}\quad	\beta= c\sum_{n \geq R} \frac{n^\zeta G(n)}{\gamma(n)}\leq \frac{c}{\gamma_1}\sum_{n\geq R}n^{\zeta+\alpha}G(n). 
	\end{equation}
\end{thm}

In keeping with the notation from the previous sections, denote by $\tilde{A} = \pi_0(A)$ the image of any observable $A\in \cA$ under the GNS representation $(\caH,\pi_0,\Omega)$ of the zero-energy ground state $\omega_0$. The proof of Theorem~\ref{thm:gen_form_bd} follows closely the argument proving \cite[Theorem 3.8]{nachtergaele:2022}, with the proviso that one must check that the infinite operator sums replacing the finite operator sums from \cite{nachtergaele:2022} are well-defined. 

To this end, let $n \geq R$, and consider the $n$-th separating partition $\mathcal{T}_n = \{ \mathcal{T}_n^i : i \in \mathcal{I}_n \}$. 
For each $i \in \mathcal{I}_n$ and any choice of $x,y \in \mathcal{T}_n^i$, one has that
\begin{equation} \label{terms_commute}
	[\tilde{H}_{\Lambda(x,n)}, \tilde{H}_{\Lambda(y,n)}] = 0 \quad \mbox{and} \quad [\tilde{P}_{\Lambda(x,n)}, \tilde{P}_{\Lambda(y,n)}] = 0 \, .
\end{equation}
This follows since the corresponding algebra elements $H_{\Lambda(x,n)},H_{\Lambda(y,n)}\in \cA$ are supported on disjoint sets (and similarly for the ground state projections), which carries over to the GNS space by the homomorphism property.

With $n \geq R$ and $i \in \mathcal{I}_n$ fixed, denote by $C_n^i$ the collection of all configurations associated to
$\mathcal{T}_n^i$. More precisely, 
\begin{equation}
	C_n^i = \{0, 1\}^{\mathcal{T}_n^i} = \left\{ \sigma : \sigma = \{ \sigma_x \} \mbox{ where } \sigma_x \in \{0, 1\} \mbox{ for all } x \in \mathcal{T}_n^i \right\} \, .
\end{equation}
For each $ \sigma \in C_n^i $, we define $|\sigma|$ by
\begin{equation}
	| \sigma | = \sum_{x \in \mathcal{T}_n^i} \sigma_x .
\end{equation}
%which is finite if and only if $|\{x \in \mathcal{T}_n^i : \sigma_x =1 \}| < \infty$.
%\begin{equation}
%| \sigma | < \infty \quad \mbox{if and only if} \quad |\{x \in \mathcal{T}_n^i : \sigma_x =1 \}| < \infty \, .
%\end{equation}

Recall that $\pi_{0}(\cA^{\rm loc}) \Omega$ is a 
dense subspace of $\mathcal{H}$. Let $\psi = \tilde{A} \Omega$ for some $A \in \mathcal{A}_X$, and set $\tilde{Q}_{\Lambda(x,n)}= \idty-\tilde{P}_{\Lambda(x,n)}$. In this case, one has that
\begin{equation}
	\tilde{P}_{\Lambda(x,n)} \psi = \psi \quad \mbox{and} \quad \tilde{Q}_{\Lambda(x,n)} \psi = 0 \, 
\end{equation}
whenever $\Lambda(x,n)$ satisfies $\Lambda(x,n) \cap X = \emptyset$. It is also clear that for 
any $X \in \mathcal{P}_0( \Gamma)$, the set of sites $\{x \in \mathcal{T}_n^i : \Lambda(x,n) \cap X \neq \emptyset \}$ is finite.
Thus, for any $\psi \in \pi_{0}(\cA^{\rm loc}) \Omega$, define the operator $S(\sigma)$ by
\begin{equation}
	S( \sigma) \psi = \left\{ \begin{array}{cl} \prod_{x \in \mathcal{T}_n^i} \left[ \sigma_x \tilde{Q}_{\Lambda(x,n)} + (1- \sigma_x) \tilde{P}_{\Lambda(x,n)} \right] \psi &
		\mbox{if } |\sigma| < \infty , \\ 0 & \mbox{otherwise}.
	\end{array} \right.
\end{equation}
Note that, if $|\sigma| < \infty$, then at most finitely many of these factors act non-trivially, and
moreover, by (\ref{terms_commute}), all factors above commute. Since $\pi_{0}(\cA^{\rm loc}) \Omega$ is
dense, there is a unique extension of $S(\sigma)$ to an element of $\mathcal{B}( \mathcal{H}_{\omega})$
for each $\sigma \in C_n^i$. One checks that these operators satisfy:
\begin{equation} \label{S_props}
	S(\sigma)^* = S( \sigma), \quad S(\sigma) S(\sigma') = \delta_{\sigma, \sigma'} S(\sigma), \quad \mbox{and} \quad
	\sum_{\sigma \in C_n^i} S(\sigma) = \idty \, .
\end{equation}

Since $\tilde{Q}_{\Lambda(x,n)}=\idty-\tilde{P}_{\Lambda(x,n)}$, by \eqref{terms_commute} one clearly has that
\begin{equation} \label{Q_P_com}
	[ \tilde{Q}_{\Lambda(x,n)}, \tilde{P}_{\Lambda(y,n)}]=0 \quad \mbox{for all } \; x,y \in \mathcal{T}_n^i \,
\end{equation}
and as a result, also 
\begin{equation} \label{Q_S_com}
	[\tilde{Q}_{\Lambda(x,n)}, S(\sigma)]=0 \quad \mbox{for each } x \in \mathcal{T}_n^i \mbox{ and all }  \;\sigma \in C_n^i.
\end{equation}
For $\sigma \in C_n^i$ with $|\sigma| < \infty$, (\ref{Q_P_com}) implies 
$0 = [ \tilde{Q}_{\Lambda(x,n)}, S(\sigma)] \psi$ for all $\psi \in \pi_{0}(\cA^{\rm loc}) \Omega$ and, hence, for all $\psi\in\cH$. If $|\sigma|=\infty$, then the commutator is zero by definition.
%
%In fact, one need only consider $\sigma \in C_n^i$ with $|\sigma| < \infty$, as otherwise the result is trivial. 
%For such configurations, one readily checks that (\ref{Q_P_com}) implies 
%$0 = [ \tilde{Q}_{\Lambda(x,n)}, S(\sigma)] \psi$ for all $\psi \in \pi_{0}(\cA^{\rm loc}) \Omega$, and this relation uniquely defines
%this bounded commutator. 
Combining (\ref{S_props}) and (\ref{Q_S_com}), we find that 
\begin{equation} \label{SQS}
	S(\sigma) \tilde{Q}_{\Lambda(x,n)} S(\sigma') = \delta_{\sigma, \sigma'} \tilde{Q}_{\Lambda(x,n)} S(\sigma) =  \delta_{\sigma, \sigma'} ( \idtyty - \tilde{P}_{\Lambda(x,n)}) S(\sigma)  =  \delta_{\sigma, \sigma'} \sigma_x S(\sigma)
\end{equation}
for all $\sigma, \sigma' \in C_n^i$ and each $x \in \mathcal{T}_n^i$.

We use these families of orthogonal projections in the following proof.

\begin{proof}[Proof of Theorem~\ref{thm:gen_form_bd}]
	In the GNS representation, the terms of $\tilde{V}$ can be rearranged using the family of separating partitions from Assumption~\ref{ass:localgaps} as
	\begin{equation}
		\tilde{V} = \sum_{x \in \Gamma} \sum_{n \geq R} \tilde{\Phi}(x,n) = \sum_{n \geq R} \sum_{i \in \mathcal{I}_n} \tilde{V}_n^i, \qquad 	\tilde{V}_n^i := \sum_{x \in \mathcal{T}_n^i} \tilde{\Phi}(x,n) 
	\end{equation}
	since $V$ is absolutely summable.	As a result, for any $\psi \in \mathcal{H}$, one has the bound
	\begin{equation} \label{basic_V_bd}
		|\langle \psi, \tilde{V} \psi \rangle | \leq \sum_{n \geq R} \sum_{i \in \mathcal{I}_n} |\langle \psi, \tilde{V}_n^i \psi \rangle | \quad \mbox{with} \quad |\langle \psi, \tilde{V}_n^i \psi \rangle | \leq \sum_{x \in \mathcal{T}_n^i} |\langle \psi, \tilde{\Phi}(x,n) \psi \rangle | \,.
	\end{equation}
	
	Now, since $b_x(n) \subset \Lambda(x,n)$ for all $\Lambda(x,n) \in \mathcal{S}$, 
	the frustration free property implies
	\begin{equation}
		P_{\Lambda(x,n)} = P_{b_x(n)} P_{\Lambda(x,n)} = P_{\Lambda(x,n)} P_{b_x(n)},
	\end{equation}
	and so by (\ref{terms_annih}) it follows that that for each $x\in\mathcal{T}_n^i$, the term $\tilde{\Phi}(x,n)$ satisfies 
	\begin{equation} \label{bits_commute}
		[ \tilde{\Phi}(x,n), \tilde{P}_{\Lambda(y,n)}]=0 \quad \mbox{for all } y \in \mathcal{T}_n^i. 
	\end{equation} 
	Arguing as in \eqref{Q_S_com}-\eqref{SQS} above, one then finds that for all $x \in \mathcal{T}_n^i$ and $\sigma,\sigma' \in C_n^i$
	\begin{equation} 
		[\tilde{\Phi}(x,n), S(\sigma)] =0 \quad \label{P_S_com}
	\end{equation}
	and, moreover,
	\begin{align}
		S(\sigma) \tilde{\Phi}(x,n) S(\sigma') &= \delta_{\sigma, \sigma'} S(\sigma) \tilde{\Phi}(x,n)  \nonumber\\
		& =  
		\delta_{\sigma, \sigma'} \sigma_x S(\sigma) \tilde{\Phi}(x,n)  \nonumber\\
		& = \delta_{\sigma, \sigma'} S(\sigma) \tilde{Q}_{\Lambda(x,n)} \tilde{\Phi}(x,n) \tilde{Q}_{\Lambda(x,n)} S(\sigma) \,. \label{Phi_to_Q}
	\end{align}
	%	In fact, the above may also be re-written as
	%	\begin{equation} \label{Phi_to_Q}
	%	S(\sigma) \tilde{\Phi}(x,n) S(\sigma') =  \delta_{\sigma, \sigma'} S(\sigma) \tilde{Q}_{\Lambda(x,n)} \tilde{\Phi}(x,n) \tilde{Q}_{\Lambda(x,n)} S(\sigma)
	%	\end{equation}
	As a consequence, one can use (\ref{S_props}) and (\ref{Phi_to_Q}) to bound
	\begin{eqnarray}
		| \langle \psi, \tilde{\Phi}(x,n) \psi \rangle | \leq \sum_{\sigma, \sigma'} | \langle \psi, S(\sigma) \tilde{\Phi}(x,n) S(\sigma') \psi \rangle | & = & 
		\sum_\sigma | \langle \psi, S(\sigma) \tilde{Q}_{\Lambda(x,n)} \tilde{\Phi}(x,n) \tilde{Q}_{\Lambda(x,n)} S(\sigma) \psi \rangle | \nonumber \\
		& \leq & G(n)  \sum_\sigma \| \tilde{Q}_{\Lambda(x,n)} S(\sigma) \psi \|^2 .
	\end{eqnarray}
	Summing these orthogonal projections again, one obtains
	\begin{equation}
		\sum_\sigma \| \tilde{Q}_{\Lambda(x,n)} S(\sigma) \psi \|^2 = \sum_\sigma \langle \psi,  \tilde{Q}_{\Lambda(x,n)} S(\sigma) \psi \rangle=
		\langle \psi,  \tilde{Q}_{\Lambda(x,n)}  \psi \rangle \leq \frac{1}{\gamma(n)} \langle \psi,  \tilde{H}_{\Lambda(x,n)}  \psi \rangle
	\end{equation}
	where, the final bound follows since  the representation 
	preserves positivity and the operator inequality $\gamma(n) Q_{\Lambda(x,n)} \leq H_{\Lambda(x,n)}$ holds in the algebra.  
	
	Given the above, one concludes that for any $\psi \in {\rm dom}(H_0)$,
	\begin{equation}
		| \langle \psi, \tilde{V}_n^i \psi \rangle| \leq \sum_{x \in \mathcal{T}_n^i} | \langle \psi, \tilde{\Phi}(x,n) \psi \rangle| \leq  \frac{G(n)}{\gamma(n)}  \sum_{x \in \mathcal{T}_n^i}  \langle \psi, \tilde{H}_{\Lambda(x,n)} \psi \rangle \leq \frac{G(n)}{\gamma(n)} \langle \psi, H_0\psi \rangle 
	\end{equation}
	where the last inequality uses that $\Lambda(x,n) \cap \Lambda(y,n) = \emptyset$ for distinct $x, y \in \mathcal{T}_n^i$. Returning to (\ref{basic_V_bd}), since these partitions satisfy the $(c,\zeta)$-polynomial growth bound it is clear that
	\begin{equation}
		| \langle \psi, \tilde{V} \psi \rangle| \leq \sum_{n \geq R} \sum_{i \in \mathcal{I}_n} | \langle \psi, \tilde{V}_n^i \psi \rangle| \leq \beta  \langle \psi, H_0 \psi \rangle \, , 
	\end{equation}
	with $\beta$ as in (\ref{V_form_bd}). This completes the proof.  
\end{proof}

\subsection{Proofs of Theorems \ref{thm:finite_volume} and \ref{thm:infinite_volume}}\label{sec:final_summary}

Theorem \ref{thm:Step2} showed that for every finite $\Lambda\subset\Gamma$, the Hamiltonian 
$$
H(\Lambda,s) = H_0 + s\pi_0(V_\Lambda) = H_0 + s \sum_{b_x(n)\subset\Lambda} \pi_0(\Phi(x,n))
$$
transforms under the spectral flow unitary as follows
$$
\tilde{U}(\Lambda,s)^*H(\Lambda,s)\tilde{U}(\Lambda,s)\psi - E(\Lambda,s)\psi = H_0\psi + \sum_{x\in\Gamma}\sum_{m\geq R}\pi_0(\Phi_\Lambda^{(2)}(x,m,s))\psi,
$$
where $E(\Lambda,s)$ is the ground state energy of $H(\Lambda,s)$, and $\Phi_\Lambda^{(2)}(x,m,s)$  is a balled interaction satisfying the conditions of Theorem~\ref{thm:gen_form_bd} with norm bounds that are linear in $|s|$ and given in terms of a $\Lambda$-dependent decay function. 
%It follows from Theorem \ref{thm:Step1} that the new perturbation of $H_0$ in the RHS is still bounded and that it is mainly supported in a neighborhood of $\Lambda$.
However, we will show below that the constant $\beta$ from Theorem~\ref{thm:gen_form_bd} can be taken independent of $\Lambda$. Theorem~\ref{thm:finite_volume} will then follow from applying \cite[Corollary 3.3]{nachtergaele:2022}. In our context, the latter result states the following. Suppose that $H_0$ is a self-adjoint, positive operator on a Hilbert space $\cH$ with $\min\spec H_0 = 0$ and $(0,\gamma_0)\cap \spec(H_0) = \emptyset$. Then, for any $V=V^*\in \cB(\cH)$ such that there exists $0\leq \beta<1$ for which
\[|\braket{\psi}{V\psi}|\leq \beta\braket{\psi}{H_0\psi} \quad \forall \psi \in \dom(H_0),\]
one has that
\begin{equation}\label{perturbed_gap}
	\spec(H_0+V) \cap \left(0,(1-\beta)\gamma_0\right) =\emptyset.
\end{equation}

\begin{proof}[Proof of Theorem \ref{thm:finite_volume}]
	Fix $\gamma \in (0,\gamma_0)$ where $\gamma_0$ is as in Assumption~\ref{ass:bulkgap}, and let $\Lambda \in \cP_0$ be arbitrary. Given Theorem~\ref{thm:Step2}, it is clear that Theorem~\ref{thm:gen_form_bd} applies and produces a non-trivial form bound (see \eqref{V_form_bd}) so long as
	\[
	G_\Lambda^{(2)}(m) := \sup_{x\in \Gamma} G_\Lambda^{(2)}(x,m)
	\] 
	has a finite $(\zeta+\alpha)$-moment where, as stated in Assumption~\ref{ass:localgaps}, $\zeta$ is the polynomial growth of the separating partitions and $\gamma(n) \geq \gamma_1 n^{-\alpha}$ is the lower bound on the local gaps. Moreover, to show that 
	\be\label{uniform_s}
	s_0(\gamma):=\inf_{\Lambda\in \cP_0} s_0^\Lambda(\gamma) >0,
	\ee
	see \eqref{gap_constraint}, we wish to show this form bound is uniform in $\Lambda$. 
	
	To this end, recall that $G_\Lambda$ from \eqref{G_Lambda} is defined in terms of two functions $G_1$, $G_2$ that are independent of $\Lambda$ and decay faster than any polynomial. It is then clear that for all $\Lambda\in\cP_0$, $x\in\Gamma$ and $m\geq R$,
	\[
	G_\Lambda(x,m) \leq G_1(m) + G_2(0,m) =: G(m),
	\]
	and similarly, $G_\Lambda^{(1)}(x,m)\leq G^{(1)}(m)$ where $G^{(1)}(m) = \sum_{m\geq R}G(m)$. It follows immediately that
	\[
	\sup_{\Lambda\in\cP_0} G_\Lambda^{(2)}(m) \leq G^{(2)}(m) := G(m/2)+2G^{(1)}(m/2) +2G^{(1)}(R)\sqrt{(1+m)^\nu G_0(m/2)}
	\]
	and that $G^{(2)}$ has a finite $(\zeta+\alpha)$-moment as long as $G_0$ satisfies \eqref{LTQO_moment}.
	
	Given the norm bound from Theorem~\ref{thm:Step2}, it follows from Theorem~\ref{thm:gen_form_bd} that for all $\psi\in \dom(H_0)$
	\[
	\braket{\psi}{W(\Lambda,s)\psi} \leq s\beta\braket{\psi}{H_0\psi} \quad \text{where} \quad \beta = 2c\sum_{m\geq R} \frac{n^{\zeta}G^{(2)}(m)}{\gamma(m)}.
	\]
	Thus, ${\rm gap}(H(\Lambda,s)) \geq \gamma_0 -s\beta \gamma_0$ by \eqref{perturbed_gap} and, in particular, \eqref{uniform_s} holds since this implies
	$$
	s_0^\Lambda(\gamma)\geq \frac{\gamma_0 - \gamma}{\beta \gamma_0}.
	$$
\end{proof}

We conclude with using the uniform estimate from Theorem~\ref{thm:finite_volume} to establish the claimed lower bound estimate on the gap of the extensively perturbed system from Theorem~\ref{thm:infinite_volume}.

\begin{proof}[Proof of Theorem \ref{thm:infinite_volume}]
	Let $0<\gamma < \gamma_0$ and consider $|s|\leq s_0(\gamma)$. Recall that for any IAS $(\Lambda_n)$, the following limits hold in $\cA$ as $n\to\infty$:
	\begin{align}
		\alpha_s^{\Lambda_n}(A) &\to \alpha_s(A), \quad \text{for all} \quad A\in\cA \\
		\delta_s^{\Lambda_n}(A) &\to \delta_s(A), \quad \text{for all} \quad A\in\cA_{\rm loc},\label{delta_s}
	\end{align}
	see \eqref{limiting_der} and \eqref{limiting_spec_flow}. As a consequence, $\omega_s^{\Lambda_n}(A) = \omega_0\circ \alpha_s^{\Lambda_n}(A) \to \omega_s(A)$ for all $A\in\cA$ and, moreover, $\omega_s$ is a ground state of $\delta_s$.
	
	%For any IAS $(\Lambda_n)$, the strong convergence in (\ref{limiting_spec_flow}) implies the 
	%existence of the state $\omega_s$ on $\cA$ in the weak-$^*$ sense:
	%\be
	%\omega_s(A) = \omega_0 (\alpha_s(A)) = \lim_n \omega_0 (\alpha_s^{\Lambda_n}
	%(A)) \quad \mbox{for any }A\in \cA \mbox{ and } s \in \mathbb{R} \, .
	%\ee
	%When  $|s|\leq s_0(\gamma)$, this state is a thermodynamic of the perturbed ground states, see (\ref{pert_vec_state}) and
	%comments following.
	%From this relation, it immediately follows that $\omega_s$ is a pure state with GNS triple
	%$(\cH,\pi_0\circ\alpha_s,\Omega)$.
	
	Now consider $A \in \cA^{\rm loc}$ for which $\omega_s(A) =0$. 
	Given Theorem \ref{thm:finite_volume}, the GNS Hamiltonian $H(\Lambda_n, s)$ along any IAS $(\Lambda_n)$ has a gap above its unique ground state lower bounded by $\gamma$ for all $|s|\leq s_0(\gamma)$. Therefore, $\omega_s^{\Lambda_n}$ satisfies \eqref{gap_check} and, in particular, the inequality
	\be
	\omega_s^{\Lambda_n}(B_n^*\delta_s^{\Lambda_n}(B_n)) \geq \gamma  \omega_s^{\Lambda_n}(B_n^*B_n)
	\label{n-gap}\ee
	holds for the observable $B_n = A - \omega_s^{\Lambda_n}(A) \idty \in \cA^{\rm loc}$. 
	%Expanding both sides above, produces
	%\[
	%\omega_s^{\Lambda_n}(A^*\delta_s^{\Lambda_n}(A)) +|\omega_s^{\Lambda_n}(A)|^2 \geq \gamma\left(\omega_s^{\Lambda_n}(A^*A)-|\omega_s^{\Lambda_n}(A)|^2 \right),
	%\]
	%from which \eq{perturbedGNSgap} follows from taking the limit as $n\to \infty$. The remaining claims follow as in  (\ref{gap_check}).
	Combining (\ref{delta_s}) and the local uniform convergence of $\alpha_s^{\Lambda_n}$, the limit $n \to \infty$
	can be taken on both sides of \eq{n-gap} to obtain \eq{perturbedGNSgap}. The remaining claims follow as in  (\ref{gap_check}).
\end{proof}

%\section*{Appendix}\label{sec:appendix}

 \section*{Acknowledgements}

Based upon work supported by the National Science Foundation under grants DMS-1813149 and DMS-21083901 (BN), and the DFG under EXC-2111--390814868 (AY). BN and AY acknowledge support of the Erwin Schrödinger International Institute for Mathematics and Physics (ESI), where part of this work was carried out during the ``Tensor Networks: Mathematical Structures and Novel Algorithms'' workshop. BN also gratefully acknowledges kind hospitality at the Technical University Munich during the final stages of this work and the Alexander von Humboldt Foundation for support provided through a Carl Friedrich von Siemens Research Award.

%\begin{thebibliography}{999}
%\bibliographystyle{amsplain}
%\providecommand{\bysame}{\leavevmode\hbox to3em{\hrulefill}\thinspace}
%\bibliography{qss}
%\end{thebibliography} 
\providecommand{\bysame}{\leavevmode\hbox to3em{\hrulefill}\thinspace}
\providecommand{\MR}{\relax\ifhmode\unskip\space\fi MR }
\providecommand{\MRhref}[2]{%
	\href{http://www.ams.org/mathscinet-getitem?mr=#1}{#2}
}
\providecommand{\href}[2]{#2}

\end{document}